\documentclass{article}

\usepackage{authblk}

\usepackage{xspace}
\usepackage[utf8]{inputenc}
\usepackage{fullpage}
\usepackage{multirow}
\usepackage{amssymb,amsmath,amsthm,amsfonts,mathtools}
\usepackage{xcolor}
\usepackage{url}
\usepackage{enumerate}

\usepackage{subcaption}
\usepackage[small]{complexity}
\usepackage{pbox}

\usepackage{hyperref}
\hypersetup{colorlinks = true}
\usepackage[capitalize,nameinlink,noabbrev]{cleveref}

\newtheorem{theorem}{Theorem}
\newtheorem{claim}{Claim}
\newtheorem{lemma}[theorem]{Lemma}

\newtheorem{proposition}[theorem]{Proposition}

\newtheorem{definition}[theorem]{Definition}

\newtheorem{corollary}[theorem]{Corollary}
\newtheorem{problem}[theorem]{Problem}

\newcommand{\shirley}[1]{\textcolor{magenta}{Shirley: #1}}

\newcommand{\fc}[3]{\ensuremath{f_{#1}^{(#2,#3)}}}
\newcommand{\fs}[2]{\ensuremath{f_{#1}^{#2}}}

\DeclareMathOperator{\Ad}{Ad}
\newcommand{\po}{\ensuremath{\overrightarrow{\chi}}\xspace}

\usepackage{tabularx}
\usepackage{nameref}
\usepackage[skins,listings]{tcolorbox}
\newcommand{\pname}[1]{\textsc{#1}}
\newtcolorbox{ProblemBox}[2][]{
	colframe=black!20!white,
	coltitle=black,
	arc=1.5mm,
	boxrule=.15mm,
	colbacktitle=white,
	colback=white,
	enhanced,
	adjusted title=flush left,
	attach boxed title to top left={yshift=-3mm,xshift=3mm},
	boxed title style={colframe=white,righttitle=-3mm,lefttitle=-3mm},
	title=#2,#1
}

\newcommand{\problemDef}[4]{
	\begin{ProblemBox}[label={#4},nameref={#1}]{\pname{#1}}
		\begin{tabularx}{\textwidth}{l l} 
			\textnormal{Input:} & \quad \textnormal{#2} \\ 
			\textnormal{Question:} & \quad \textnormal{#3}
		\end{tabularx}
	\end{ProblemBox}
}

\newcommand{\paraProblemDef}[5]{
	\begin{ProblemBox}[label={#5},nameref={#1}]{\pname{#1}}
		\begin{tabularx}{\textwidth}{l l} 
			\textnormal{Input:} & \quad \textnormal{#2} \\ 
			\textnormal{Parameter:} & \quad \textnormal{#3}. \\ 
			\textnormal{Question:} & \quad \textnormal{#4}
		\end{tabularx}
	\end{ProblemBox}
}

\newclass{\Hard}{hard\ }
\newclass{\Hness}{hardness\ }

\newclass{\Complete}{complete\ }
\newclass{\Cness}{completeness\ }
\newcommand{\NPc}{\NP\text{-}\Complete}

\usepackage{tikz}
\usepackage{txfonts}
\usetikzlibrary{arrows, shapes,backgrounds,arrows.meta, shapes, positioning}

\newenvironment{claimproof}[1]{\par\noindent\underline{Proof:}\space#1}{\hfill $\blacksquare$}

\DeclareMathOperator{\pot}{pot}

\begin{document}

\title{On the proper orientation number of chordal graphs\thanks{This work was supported by CNPq projects 437841/2018-9, 304478/2018-0, 311679/2018-8 and 421660/2016-3, CAPES/STIC-AmSud 88881.197438/2018-01 and FAPEMIG APQ-02592-16.}}
\author[1]{Julio Araujo\thanks{Corresponding author.}}
\author[1]{Alexandre Cezar}
\author[2]{Carlos V.G.C. Lima}
\author[3]{Vinicius F. dos Santos}
\author[1]{Ana Silva}
\affil[1]{ParGO, Departamento de Matem\'atica, Universidade Federal do Cear\'a, Brazil\par \texttt{\{julio, anasilva\}@mat.ufc.br} and \texttt{alexcezar@alu.ufc.br}}
\affil[2]{Centro de Ci\^encias e Tecnologia, Universidade Federal do Cariri, Brazil\par \texttt{vinicius.lima@ufca.edu.br}}
\affil[3]{Departamento de Ciência da Computação, Universidade Federal de Minas Gerais, Brazil\par\texttt{viniciussantos@dcc.ufmg.br}}
\maketitle
\begin{abstract}
An orientation $D$ of a graph $G=(V,E)$ is a digraph obtained from~$G$ by replacing each edge by exactly one of the two possible arcs with the same end vertices.
For each~$v \in V(G)$, the indegree of~$v$ in~$D$, denoted by~$d^-_D(v)$, is the number of arcs with head~$v$ in~$D$.
An orientation~$D$ of~$G$ is proper if~$d^-_D(u)\neq d^-_D(v)$, for all~$uv\in E(G)$.
An orientation with maximum indegree at most~$k$ is called a $k$-orientation. 
The proper orientation number of~$G$, denoted by~$\po(G)$, is the minimum integer~$k$ such that~$G$ admits a proper $k$-orientation.
We prove that determining whether~$\po(G) \leq k$ is \NPc for chordal graphs of bounded diameter, but can be solved in linear-time in the subclass of quasi-threshold graphs. When parameterizing by $k$, we argue that this problem is \FPT~for chordal graphs and argue that no polynomial kernel exists, unless \NP~$\subseteq\coNP/\poly$. We present a better kernel to the subclass of split graphs and a linear kernel to the class of cobipartite graphs.

Concerning bounds, we prove tight upper bounds for subclasses of block graphs. We also present new families of trees having proper orientation number at most 2 and at most 3. Actually, we prove a general bound stating that any graph $G$ having no adjacent vertices of degree at least $c+1$ have proper orientation number at most $c$. This implies new classes of (outer)planar graphs with bounded proper orientation number. We also prove that maximal outerplanar graphs $G$ whose weak-dual is a path satisfy $\po(G)\leq 13$. Finally, we present simple  bounds to the classes of chordal claw-free graphs and cographs.
\end{abstract}

\definecolor{dark-red}{rgb}{0.7,0.15,0.15}
\definecolor{light-gray}{gray}{0.3}
\hypersetup{citecolor = dark-red, linkcolor = light-gray, breaklinks = true}


\section{Introduction}
\label{sec:intro}

For basic notions and notations on Graph Theory and (Parameterized) Computational Complexity, the reader is referred to~\cite{BoMu08, garey_johnson, Cyganetal2015, FLSZ2019}.
All graphs in this work are considered to be finite and simple. 

An {\it orientation}~$D$ of a graph $G=(V,E)$ is a digraph obtained from~$G$ by replacing each edge by exactly one of the two possible arcs with the same endvertices.
For each $v \in V(G)$, the \emph{indegree} of~$v$ in~$D$, denoted  by~$d^-_D(v)$, is the number of arcs with head~$v$ in~$D$.
We use the notation $d^-(v)$ when the orientation $D$ is clear from the context.
An orientation $D$ of $G$ is \emph{proper} if $d^-(u)\neq d^-(v)$, for all $uv\in E(G)$.
An orientation with maximum indegree at most~$k$ is called a \emph{$k$-orientation}. 
The \emph{proper orientation number} of a graph $G$, denoted by~$\po(G)$, is the minimum integer~$k$ such that~$G$ admits a proper $k$-orientation. We say that a proper orientation $D$ of a graph $G$ is \emph{optimal} if~$\po(G) = \max_{v\in V(G)} d_D^-(v)$.

This graph parameter was introduced by Ahadi~and~Dehghan~\cite{Ahadi2013}. 
They observed that it is well defined for any graph~$G$, since one can always inductively obtain a proper $\Delta(G)$-orientation by removing a vertex of maximum degree.
Note that every proper $k$-orientation $D$ of a graph~$G$ defines a proper $(k+1)$-coloring $c:V(G)\to [k+1]$ such that $c(v) = d^-_D(v)+1$, for each $v\in V(D)$.
Hence, we have:
\begin{equation}
    \omega(G) - 1 \ \leq \ \chi(G)-1 \ \leq \ \po(G) \ \leq \ \Delta(G).
\end{equation}

With respect to Computational Complexity, given a graph $G$ and a positive integer $k$, deciding whether ${\po(G)\leq k}$ is \NP-complete, even if $G$ is the line graph of a regular graph~\cite{Ahadi2013}; or if $G$ is a planar subcubic graph~\cite{ACD+15}; or if $G$ is planar bipartite with maximum degree~5~\cite{ACD+15}. Observe that the above equation together with the latter two results imply that the problem is para-\NP-complete even on planar graphs when parameterized by $k$. Our first contribution in \cref{sec:chordal} is to prove that this problem is still \NP-complete even if restricted to chordal graphs. Our reduction also implies that no polynomial kernel to the problem restricted to chordal graphs can exist, unless~{\NP~$\subseteq$ \coNP/\poly}.

In~\cite{ALSSS19}, the authors generalize the notion of proper orientations for edge-weighted graphs and present some complexity results for trees and graphs of bounded treewidth.
In particular, they show an algorithm that decides whether $\po(G)\leq k$ when~$G$ has treewidth at most $t$ whose complexity is $\mathcal{O}(2^{t^2}\cdot k^{3t}\cdot t \cdot n)$.
By using  this previous result, in \cref{sec:parameterized} we argue that this problem admits a fixed parameter tractable (\FPT) algorithm parameterized by $k$ when $G$ is chordal. This is a contrast with the general para-\NP-completeness of the problem. This also implies on the existence of kernel for the problem, even though, as previously said, such a kernel cannot be of polynomial size unless~{\NP~$\subseteq$ \coNP/\poly}.  We then present better kernels for split graphs and cobipartite graphs, although the general computational complexity for these classes remain open. 
In \cref{sec:quasi_threshold}, we also show that if~$G$ is a quasi-threshold graph~$G$, then~$\po(G)=\omega(G)-1$ and thus one can compute~$\po(G)$ in linear time, since $G$ is chordal and $\omega(G)$ can be computed in linear time by a Lex-BFS~\cite{rose1976algorithmic}. 

Concerning bounds, in \cref{sec:split}, we first prove that if $G$ is a split graph, then $\po(G)\leq 2\omega(G)-2$ and that this bound is tight, thus answering Problem 9(a) presented in~\cite{ACD+15}. We would like to emphasize that Problem 9(b)~\cite{ACD+15} has a trivial answer as, for any cobipartite graph $G$, we have that $\po(G)\leq \Delta(G)\leq 2\omega(G)-1$. 
Also in~\cite{ACD+15} the authors prove that $\po(T)\leq 4$ for a tree $T$ and that this bound is tight. A simpler proof of this fact was provided by~\cite{Knoxetal2017}. This result has been generalized to the weighted case in~\cite{ALSSS19}. In~\cite{AHLS16}, the authors prove that the more general class of cacti graphs $G$ satisfy $\po(G)\leq 7$ and prove that this bound is tight. Their tight example is actually more general. We say that a block graph $G$ is \emph{$k$-uniform} if every clique of $G$ has size $k$. In~\cite{AHLS16}, the authors prove that, for every integer $k \geq 2$, there exists a $k$-uniform block graph $G$ such that~$\po(G)\geq 3k-2$. Note that such example is tight for trees ($k=2$) and cacti ($k=3$). 
In \cref{sec:block}, we prove that their bound is tight, i.e., we prove that, for every $k\ge 2$, if $G$ is a $k$-uniform block graph, then $\po(G)\leq 3k-2$. Note that, in particular, this implies the bound of~\cite{ACD+15} for trees.  We also show that if $G$ is a $k$-uniform block graph in which each maximal clique has at most two cut-vertices and $k\ge 3$, then $\po(G)\leq \omega(G)+1$; moreover, this bound is tight. 
Note that for $k=2$, this subclass of block graphs corresponds to the ones of trees and there exist trees with proper orientation number equal to~4~\cite{ACD+15}.

Now, observe that the previously mentioned algorithm on graphs of bounded treewidth~\cite{ALSSS19}, combined with the fact that cacti have treewidth at most~2, imply that there is a polynomial-time algorithm to compute $\po(G)$, when $G$ is a cactus graph. However, there is so far no characterization even of trees $T$ satisfying $\po(T) = i$, for $i\in\{2,3,4\}$. Some partial results are presented in~\cite{ACD+15}. Still in \cref{sec:block}, we provide a general bound stating that if $G$ has no adjacent vertices of degree at least~$c+1$, then $\po(G)\leq c$. This result implies that if $T$ is a tree with no adjacent vertices of degree at least $i+1$, then $\po(T)\leq i$, for $i\in\{2,3\}$.

A major question regarding the proper orientation number is whether there exists a positive integer~$c$ such that $\po(G)\leq c$, for any (outer)planar graph $G$. This seems a challenging problem, and has been previously asked in~\cite{ACD+15}. 
Some upper bounds on subclasses of planar graphs have been provided in~\cite{AHLS16, Knoxetal2017}. Recently, Ai et al.~\cite{Aietal2020} presented upper bounds for subclasses of triangle-free outerplanar graphs. The previous result also implies that (outer)planar graphs having no adjacent vertices of large degree have bounded proper orientation number. In \cref{sec:maximal_outerplanar}, we also prove that a maximal outerplane graph~$G$ whose weak-dual graph is a path satisfies~$\po(G)\leq 13$. Observe that, since outerplanar graphs have treewidth at most~2, the proper orientation number of this subclass of outerplanar graphs can also be computed in polynomial time by the algorithm presented in~\cite{ALSSS19}.

In~\cite{AHLS16}, the authors observe that claw-free planar graphs have bounded maximum degree, and thus bounded proper orientation number. In \cref{sec:claw_free}, we observe that  claw-free chordal graphs $G$ also has bounded maximum degree, and as a consequence we get that $\po(G)\leq \Delta(G)\leq 3\omega(G)$.

It is well known that the class of quasi-threshold graphs is the intersection of cographs and interval graphs. To the more general class of cographs, we present a simple tight upper bound in \cref{sec:cographs}. Note that this class is not a subclass of chordal graphs. Finally, in \cref{sec:further} we present several related open problems for further research.

For completeness, we emphasize that other closely related weighted parameters have also been defined and studied in the literature~\cite{Ahadi2017, DH2020}.

In the remainder of the text, we use the following notation. Given disjoint vertex sets $A$ and $B$ of a graph~$G$, we denote by~$[A, B]$ the set of edges between $A$ and~$B$ in~$G$.
The neighborhood of a vertex~$v$ in a subgraph~$H$ of~$G$ is denoted by~$N_H(v)$, that is, the set of neighbors of~$v$ that are in~$V(H)$.
The degree of~$v$ in~$H$, is denoted by~$d_H(v)$. 
We omit the subscript when it is clear from the context.

\section{Computational Complexity}
\label{sec:CompCompl}

\subsection{NP-completeness of \pname{Proper Orientation} on chordal graphs with bounded diameter}
\label{sec:chordal}

A \emph{chordal graph} is a graph in which every cycle of length at least~4 has a \emph{chord}, that is, an edge joining two non-consecutive vertices of the cycle.
Let us consider the computational complexity of the following decision problem.

\problemDef{Proper Orientation}{A graph $G$ and a positive integer $k$.}{Is $\po(G) \leq k$?}{pro:PO}

We prove that \nameref{pro:PO} is \NP-complete for chordal graphs of bounded diameter. 
This result is a reduction from the well known \pname{Vertex Cover} problem on cubic graphs, which is \NP-complete~\cite{garey_johnson}.
A \emph{vertex cover} of a graph $G$ is a vertex set~$D \in V(G)$ such that every edge of $G$ has an endvertex in~$D$. 
The \emph{vertex cover number} $\tau(G)$ is the size of a minimum vertex cover of~$G$.


\problemDef{Vertex Cover}{A graph $G$ and a positive integer $k$.}{Is $\tau(G) \leq k$?}{pro:VC}

We first present a set of gadgets that we use in our reduction. 
We start by providing a structure that, given a vertex $v$ and some $i\in[k]$, prevents $v$ from having indegree $i$ in every proper $k$-orientation $D$ of the input graph~$G$.
For this, we build a gadget such that $v$ is adjacent to a vertex $u$, and we ensure that $u$ has indegree $i$ and that the edge~$uv$ must be oriented towards $u$ in any proper $k$-orientation of $G$. 
In order to build such a structure, we define below another gadget $S(k)$ containing vertices $v_0,\dotsc, v_k$, and such that $v_j$ must have indegree $j$ in every proper $k$-orientation of $G$, for each $j \in \{0\} \cup [k]$. 
\cref{fig:Sik} depicts $S(4)$.

\begin{figure}[b!]
	\centering
	\begin{subfigure}[b]{0.35\textwidth}
	    \centering
    	\begin{tikzpicture}
    	
    	    \tikzset{node/.style={draw, fill=white, circle, thick, fill = white, auto=left, inner sep=2pt}}
            \tikzset{rectangle/.append style={draw,rounded corners,dashed, minimum height=0.4cm}}
            
    	    \node[node] (v0) at (0,0) [fill=gray, label=left:$v_0$] {};
    	    \node[node] (n01) at (1,0) {};
    	    \node[node] (n02) at (2,0) {};
		    \node[node] (n03) at (3,0) {};
	    	\node[node] (n04) at (4,0) {};
    		\node[rectangle] (r0) at (2.05,0) [minimum width=4.5cm, label=right:$K^{0}$] {};
    	
		    \node[rectangle] (r1) at (1.55,-1) [minimum width=3.5cm, label=right:$K^{1}$] {};
		
		    \node[node] (n10)  at (0,-1) [fill=gray, label=left:$v_1$] {};
	    	\node[node] (n11)  at (1,-1) {};
    		\node[node] (n12)  at (2,-1) {};
		    \node[node] (n11)  at (3,-1) {};
		
		    \node[rectangle] (r2) at (1.05,-2) [minimum width=2.5cm, label=right:$K^{2}$] {};
		
		    \node[node] (n20)  at (0,-2) [fill=gray, label=left:$v_2$] {};
	    	\node[node] (n21)  at (1,-2) {};
    		\node[node] (n22)  at (2,-2) {};
    		
    		\node[rectangle] (r3) at (0.55,-3) [minimum width=1.5cm, label=right:$K^{3}$] {};
		
		    \node[node] (n30)  at (0,-3) [fill=gray, label=left:$v_3$] {};
		    \node[node] (n30)  at (1,-3) {};
		    
		    \node[node] (n40)  at (0,-4) [fill=gray, label=left:$v_4$] {};
		
		    \node[rectangle] (r3) at (0,-2) [minimum width=4.5cm, rotate=90, label=left:$K$] {};
		    
	    \end{tikzpicture}
	    \caption{Graph $S(4)$.}
	    \label{fig:Sik}
	\end{subfigure} \qquad
	~
	\begin{subfigure}[b]{0.35\textwidth}
	    \centering
		\begin{tikzpicture}[level/.style={sibling distance=20mm/#1}]
		
		    \tikzset{node/.style={draw, fill=white, circle, thick, fill = white, auto=left,
            inner sep=2pt}}
            \tikzset{>={Latex[width=2mm,length=2mm]}}

            \node[node] (u) at (0,0.3) [label=above:$u$] {};
            \node[node] (u') at (2,0.3) [label=above:$u'$] {};
	
            \node[node] (v0) at (1,-1.2) [fill=gray, label=below:$v_0$] {};
            \node[node] (v1) at (1.5,-1.2) [fill=gray, label=below:$v_1$] {};
	        \node[node] (vi1) at (2.5,-1.2) [fill=gray, label=below:$v_{i-1}$] {};
	        \node[node] (vi) at (3,-1.2) [fill=gray, label=below:$v_{i}$] {};
	        \node[node] (vk) at (4,-1.2) [fill=gray, label=below:$v_{k}$] {};
	    
	        \path[->] (u) edge node {} (u')
                (v1) edge node {} (u')
                (vi1) edge node {} (u');
                
            \path (v1) -- (vi1) node (x) [midway] {$\dotsc$};
            \path (vi) -- (vk) node (y) [midway] {$\dotsc$};
            \node (l) at (2.4,-2.7) [label=$S(k)$] {};
	    
	        \draw (2.4,-1.5) ellipse (54pt and 28pt) {};
	        
        \end{tikzpicture}
		\caption{Operation to forbid color $i$ on $u$.}
		\label{fig:Fik}
	\end{subfigure}
	
	\caption{In \cref{fig:Sik}, each dashed rectangle represents a clique.
	For simplicity, we omitted the edges from $v_j$ to the cliques~$K^\ell$, for every $j\in\{0\}\cup [2]$ and every $j+1\leq \ell\leq 3$.
	In \cref{fig:Fik} $u'$ must have indegree $i$, since it is adjacent to~$v_1, v_2, \dotsc, v_{i-1}$ of~$S(k)$, and $uu'$ must be oriented towards $u'$.}
	\label{fig:}
\end{figure}
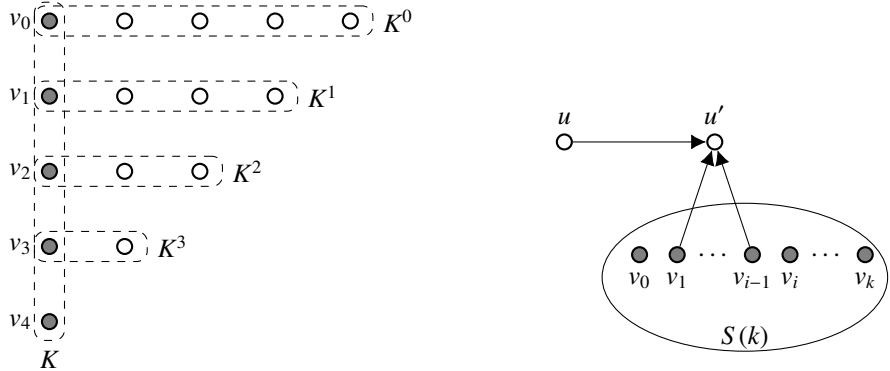

\begin{definition}
\label{def:Sk}
Given a non-negative integer $k$, let $S(k)$ be constructed as follows. First add to $S(k)$ a clique~$K$ on~$k+1$ vertices $\{v_0,\ldots, v_k\}$.
Then, for each $j\in\{0\}\cup [k-1]$, add to $S(k)$ a clique $K^j$ on~$(k+1-j)$ vertices such that the only common vertex of $K^j$ with other vertices in $S(k)$ is $v_j$.
Finally, add edges from~$v_j$ to all vertices of $K^\ell$, for every $j\in\{0\}\cup [k-2]$ and every $\ell\in\{j+1,\cdots, k-1\}$. 
\end{definition}

\begin{proposition}\label{prop:Skchordal}
Given a non-negative integer $k$, let $S(k)$ the graph constructed as in Definition~\ref{def:Sk}. Then $S(k)$ is a chordal graph.
\end{proposition}
\begin{proof}
By induction on $k$. Clearly $S(0)$ is chordal since it consists of a single vertex, $v_0$. Now, observe that if we remove $v_k$ from $K$, and exactly one vertex $x_j$ of $K^j\setminus \{v_j\}$ for each $j\in\{0\}\cup[k-1]$, then we obtain graph $S(k-1)$. By induction hypothesis, let $\pi_{k-1}$ be a perfect elimination ordering of $S(k-1)$. We claim that each vertex in $X = \{x_0,\cdots,x_{k-1},v_k\}$ is a simplicial vertex in $S(k)$. Indeed, because $N(v_k) = K$ and $N(x_0) = K^0$, we get that $v_k$ and $x_0$ are simplicial. Now, consider $j\in \{1,\cdots, k-1\}$, and note that $N(x_j) = K^j\cup \{v_0,\cdots,v_{j-1}\}$. By construction, we know that $v_{j'}x$ is an edge of $S(k)$ for every $x\in K^j$ and every $0\le j'<j$. It follows that $N(x_j)$ is indeed a clique. 
We then get that the ordering obtained from $\pi_{k-1}$ by adding $X$ at the beginning is a perfect elimination ordering of $S(k)$.
\end{proof}


\begin{proposition}\label{prop:Sk}
For a positive integer $k$, let $G$ be a graph such that $S(k)$ is an induced subgraph of $G$ and let $K$ be the clique $\{v_0, \dotsc, v_k\}$, as presented in \cref{def:Sk}.
Then, in any proper $k$-orientation $D$ of $G$, we have that $d^-_D(v_j) = j$, for each $v_j \in K$.
Moreover, for every edge $uv \in E(G)$ such that $v\in S(k)$ and $u \in V(G) \setminus V(S(k))$, we have that $uv$ must be oriented towards $u$.
\end{proposition}
\begin{proof}
Consider a proper $k$-orientation $D$ of $G$. One should first observe that, if $K'$ is any clique on~$k+1$ vertices in $G$, then we get that $D$ restricted to $K'$ must be transitive and we must have exactly one vertex in $K'$ with indegree $j$, for every $j\in\{0,\dotsc, k\}$. 
This implies that any edge $uv$ between a vertex $v\in K'$ and a vertex $u\notin K'$ must be oriented towards $u$. 
Now, observe that~$v_0$ is the intersection of two cliques $K$ and $K^0$ of cardinality $k+1$ in $S(k)$ (see \cref{fig:Sik}). 
We then get that in any proper $k$-orientation $D$ of $G$ we must have $d^-_D(v_0) = 0$.

Since $v_0$ is adjacent to all vertices in $K^1$, which is a clique of size $k$ by \cref{def:Sk}, we get that the vertices in~$K^1$ must have indegree from $1$ to $k$, i.e. there is exactly one vertex in~$K^1$ with indegree $j$, for each $j\in [k]$. 
Moreover, since $v_1\in K$, then if~$u\in K^1\setminus \{v_1\}$, we have that $uv_1$ is oriented towards $u$.
Consequently, $d^-_D(v_1) = 1$.
Note that the orientation of $K^1$ is also transitive and every edge $uv \in E(G)$, such that~$v\in K^1$ and $u \in V(G) \setminus (K^1\cup\{v_0\})$, must be oriented towards $u$.

Applying this argument inductively, we deduce that $d^-_D(v_j) = j$, for each $j\in \{2,\dotsc, k\}$, and any edge~$uv$ with $v\in S(k)$ and $u\notin S(k)$ must be oriented towards $u$.
\end{proof}

\cref{prop:Sk} shows that, in any proper $k$-orientation of a graph $G$ having $S(k)$ as a subgraph and having a vertex $u \in V(G) \setminus V(S(k))$ that is adjacent to $v_j$, for some $j \in \{0\} \cup [k]$, the indegree of $u$ must be distinct from~$j$.
Moreover we also have that the edge $uv_j$ must be oriented towards $u$. 
This allows us to be able to forbid certain indegrees for $u$, but with the disadvantage of increasing the indegree of $u$. 
%
In order to avoid increasing the indegree, we construct the following additional gadget (See \cref{fig:Fik}).

\begin{definition}
\label{def:Fi}
Given a positive integer $k$ and $i \in [k] \setminus\{1\}$, let $F(i,k)$ be the graph obtained from $S(k)$ by adding a vertex $u'$, called the \emph{head} of $F(i,k)$, such that $u'$ is adjacent to the vertices $v_1,\dotsc, v_{i-1}$ in the clique $K$ of $S(k)$.
\end{definition}

Observe that, since $F(i,k)$ is obtained from $S(k)$ by adding a simplicial vertex, it follows by Proposition~\ref{prop:Skchordal} that $F(i,k)$ is also a chordal graph. 
We say that a graph $G$ has a \emph{pendant $F(i,k)$ at $u \in V(G)$} if $G$ has $F(i,k)$ as an induced subgraph and the only edge between $V(F(i,k))$ and $ V(G)\setminus V(F(i,k))$ is~$uu'$, where $u'$ is the head of $F(i,k)$.

\begin{proposition}
\label{prop:Fik}
Let $G$ be a graph having a pendant $F(i,k)$ at some $u \in V(G)$, with $2\leq i \leq k$.
Then, every proper $k$-orientation $D$ of $G$ is such that $d^-_D(u') = i$ and $uu'$ is oriented towards $u'$.
\end{proposition}
\begin{proof}
Let $D$ be any proper $k$-orientation of $G$. By \cref{prop:Sk}, we have that the edges $v_ju'$ of $F(i,k)$ (see \cref{fig:Fik}) must be oriented towards $u'$, for every $j \in [i-1]$, and thus $d^-_D(u')\geq i-1$. 
Since $d^-_D(v_{i-1}) = i-1$, we deduce that the edge $uu'$ must be oriented towards $u'$, which implies that $d^-_D(u')=d_G(u')=i$.
\end{proof}

Now observe that \cref{prop:Fik} allows us to forbid indegree $d\in\{2,\ldots,k\}$ to a vertex $u$ having a pendant~$F(d,k)$ attached to it.
Thus, the above construction enables us to associate a list assignment~$L: V(G)~\to 2^{\{0\}\cup [k]}$ of indegrees such that $u$ must have indegree in~$L(u)$ in a proper $k$-orientation of~$G$, for every $u\in V(G)$. 
Observe that, since $k \ge 2$, \cref{prop:Fik} cannot be used to forbid the values 0 and 1 from appearing in the list of a vertex. 

We are now ready to present our reduction. Recall that a \emph{split graph} is a graph whose vertex set can be partitioned into a clique and an independent set. Notice that a split graph is a subclass of chordal graphs.

\begin{figure}[b!]
	\centering
    \begin{tikzpicture}
    	
    	\tikzset{node/.style={draw, fill=white, circle, thick, fill = white, auto=left, inner sep=2pt}}
    	\tikzset{>={Latex[width=2mm,length=2mm]}}
	
        \node[node] (u) at (1,.75) {$v$};
        \node[node] (w1) at (1,2.5) {};
        \node[node] (w2) at (1,-1) {};
        
        \draw (u) -- (w1) {};
        \draw (w2) to [bend right] (w1) {};
        \path (u) -- (w2) node (r') [midway,rotate=90] {$\dots$};
        \draw (u) -- (r') {};
        \draw (r') -- (w2) {};
        
        \draw[dashed] (1,1) ellipse (26pt and 70pt) {};
        \node (K) at (1,-1.5) [label=below:$\mathcal{K}$] {};
	    
        \node[node] (I2) at (3.5,.75) {$u$};
        \node[node] (I1) at (3.5,2.5) {};
        \node[node] (I3) at (3.5,-1) {};
        
        \draw (I2) -- (w1) {};
        \path (I2) -- (I3) node (r'') [midway,rotate=90] {$\dots$};
        
        \draw[dashed] (3.5,1) ellipse (26pt and 70pt) {};
        \node (I) at (3.5,-1.5) [label=below:$\mathcal{I}$] {};
	    
	    \node[node] (u'1) at (-.5,3) {};
	    
        \node[node] (v01) at (-2,4) [fill=gray, label=left:$v_0$] {};
        \node[node] (v11) at (-2,3.5) [fill=gray, label=left:$v_1$] {};
	    \node[node] (vk11) at (-2,2.75) [fill=gray, label=left:$v_{k-1}$] {};
	    \node[node] (vk'1) at (-2,2) [fill=gray, label=left:$v_{k'}$] {};
	    
	    \path[->] (u) edge node {} (u'1)
            (v11) edge node {} (u'1)
            (vk11) edge node {} (u'1);
                
        \path (v11) -- (vk11) node (x1) [midway,rotate=90] {$\dots$};
        \path (vk11) -- (vk'1) node (y1) [midway,rotate=90] {$\dots$};
        \node (Fkk') at (-3,3) [label=left:$F(k{,\,}k')$] {};
	    
	    \draw (-2.2,3) ellipse (26pt and 44pt) {};
	    
        
        \node[node] (u'2) at (-.5,-1.5) {};
        
        \node[node] (v02) at (-2,-.5) [fill=gray, label=left:$v_0$] {};
        \node[node] (v12) at (-2,-1) [fill=gray, label=left:$v_1$] {};
	    \node[node] (vk22) at (-2,-1.75) [fill=gray, label=left:$v_{k}$] {};
	    \node[node] (vk'2) at (-2,-2.5) [fill=gray, label=left:$v_{k'}$] {};
	    
	    \path[->] (u) edge node {} (u'2)
            (v12) edge node {} (u'2)
            (vk22) edge node {} (u'2);
                
        \path (v12) -- (vk22) node (x2) [midway,rotate=90] {$\dots$};
        \path (vk22) -- (vk'2) node (y2) [midway,rotate=90] {$\dots$};
        \node (Fk3k') at (-3,-1.5) [label=left:$F(k+1{,\,}k')$] {};
	    
	    \draw (-2.2,-1.5) ellipse (26pt and 44pt) {};
        
	    \node[node] (u'I) at (5,3) {};
	    
        \node[node] (vI0) at (6.5,4) [fill=gray, label=right:$v_0$] {};
        \node[node] (vI1) at (6.5,3.5) [fill=gray, label=right:$v_1$] {};
	    \node[node] (vIk2) at (6.5,2.75) [fill=gray, label=right:$v_{k-2}$] {};
	    \node[node] (vk'I) at (6.5,2) [fill=gray, label=right:$v_{k'}$] {};
	    
	    \path[->] (I2) edge node {} (u'I)
            (vI1) edge node {} (u'I)
            (vIk2) edge node {} (u'I);
                
        \path (vI1) -- (vIk2) node (x1) [midway,rotate=90] {$\dots$};
        \path (vIk2) -- (vk'I) node (y1) [midway,rotate=90] {$\dots$};
        \node (Fk1k') at (9.8,2) [label=left:$F(k-1{,\,}k')$] {};
	    
	    \draw (6.7,3) ellipse (26pt and 44pt) {};

        \node[node] (s1) at (5.3,0.9) [label=right:$s(Z_1)$] {};
        \node[node] (s2) at (5.3,-1.8) [label=right:$s(Z_{k-1})$] {};
        \node[rotate=90] (Zr) at (5.45,-0.4) {$\dots$};
        
        \path[rotate=15] (s1) {} edge [in=-50,out=-130,loop,min distance=18mm] node[label=92.5:$K_{k'}$] {} (s1);
        \path[rotate=-15] (s1) {} edge [in=50,out=130,loop,min distance=18mm] node[label=-92.5:$K_{k'}$] {} (s1);

        \path[rotate=15] (s2) {} edge [in=-50,out=-130,loop,min distance=18mm] node[label=92.5:$K_{k'}$] {} (s2);
        \path[rotate=-15] (s2) {} edge [in=50,out=130,loop,min distance=18mm] node[label=-92.5:$K_{k'}$] {} (s2);
        
        \path[->] (s1) edge node {} (I2)
            (s2) edge node {} (I2);

	\end{tikzpicture}
	\caption{Graph $G'$ in the reduction of \cref{thm:triviallyPG}.
	The remaining gadgets in the vertices of $\mathcal{K}$ and $\mathcal{I}$, as well as the edges from $\mathcal{K}$ to $\mathcal{I}$, are omitted for simplicity.}
	\label{fig:reduction}
\end{figure}
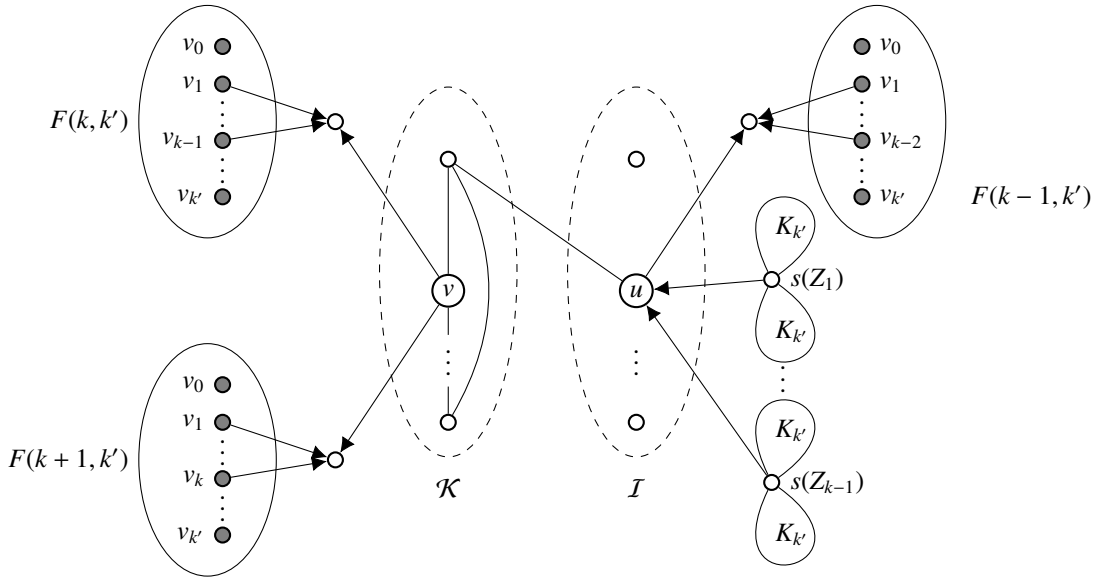

\begin{theorem}\label{thm:NP-chordal}
\textsc{\nameref{pro:PO}} is \NP-complete on chordal graphs of bounded diameter.
\end{theorem}
\begin{proof}
\nameref{pro:PO} is already known to be in \NP~\cite{Ahadi2013}.
Let~$(G, k)$ be an instance of \nameref{pro:VC}, such that~$G$ is a cubic graph. We may assume that $k\geq 2$, as otherwise \nameref{pro:VC} can be solved in linear time, and that $G$ is connected.


Let $k' = |V(G)|+2$.
In the sequel, we construct an instance $(G',k')$ of \nameref{pro:PO} from $(G, k)$, such that $G'$ is a chordal graph of bounded diameter and~$G$ has a vertex cover of size at most $k$ if and only if~$G'$ admits a proper~$k'$-orientation~$D$.

Let $H=(\mathcal{K} \cup \mathcal{I}, E_H)$ be the split graph obtained from~$G=(V, E)$, such that $\mathcal{K} = V(G)$ is a clique, $\mathcal{I}=E(G)$ is an independent set and we add an edge between the vertex $uv \in \mathcal{I}$ to only~$u\in \mathcal{K}$ and $v\in \mathcal{K}$.
Therefore, $d_{\mathcal{I}}(v) = 3$ for each $v \in \mathcal{K}$, since $G$ is cubic; and $d_{\mathcal{K}}(uv) = 2$, for each $uv \in \mathcal{I}$.

We continue the construction of $G'$ from $H$, by using the gadgets $F(i,k')$. To each vertex $v \in {\cal K}$, we add a pendant $F(i,k')$ at each vertex, for each $i\in\{k,k+1\}$ (see \cref{fig:reduction}). Moreover, we add a pendant~$F(k-1,k')$ at $u$, for each $u \in \mathcal{I}$. Finally let $Z$ be the graph obtained by two cliques of size~$k'$ sharing a unique vertex~$s(Z)$.
For each~$u \in \mathcal{I}$, we add $k-1$ distinct subgraphs $Z_i$ and connect $u$ to $s(Z_i)$, for each $i \in [k-1]$.
This completes the construction of $G'$.


Since we obtain a split graph $H$ from $G$ and add to it chordal subgraphs $F(i,k')$ and $Z$ attached by a cut-vertex that belongs to no chorless cycles, then one may observe that $G'$ is chordal.
Moreover, since the diameter of any split graph is at most three, we can also deduce that the diameter of $G'$ is at most nine (in \cref{fig:reduction}, take the distance between a vertex in~$K^i \setminus K$ from an~$F(k-1,k')$ attached to~$u \in \mathcal{I}$ to another similar vertex in another~$F(k-1,k')$ attached to~$u' \in \mathcal{I}$ and such that $N_{\mathcal{K}}(u)\cap N_{\mathcal{K}}(u') = \emptyset$).

Let us now prove that if $G$ has a vertex cover $S$ of size at most $k$, then~$G'$ admits a proper~$k'$-orientation~$D$.
Let $S \subseteq V(G)$ be a vertex cover of $G$ of size \emph{exactly} $k$ (we add arbitrary extra vertices, if necessary). Let us consider the same set $S\subseteq {\cal K}$ in $G'$. Since $S$ is a vertex cover of $G$, recall that each vertex of $\mathcal{I}$ is adjacent to at least one vertex of $S$.

We construct a proper $k'$-orientation $D$ of $G'$ as follows.
First, add to $D$ any proper $k'$-orientation of each pendant subgraph $F(i,k')$, for $i\in\{k-1,k+1\}$. Recall that they must verify Proposition~\ref{prop:Fik}. Similarly, add to $D$ any proper $k'$-orientation of each subgraph corresponding to $Z$. Note that $s(Z)$ must be a source in $D$, to each such subgraph. It already implies that $d_D(v)\geq k-1$, to each vertex $v\in {\cal I}$.
It suffices now to orient the edges with both endpoints in $H$.
Now, take an arbitrary transitive orientation $G'[S]$.
For each edge $vu$ such that $v\in S$ and $u\in H\setminus S$, orient $vu$ towards~$u$.
Then, observe that $d^-_D(v) \in \{0\} \cup [k-1]$, for every $v \in S$. Moreover, each vertex $u\in\mathcal{I}$ is the head of at least one arc with tail in $S$, as $S$ is a vertex cover of $G$.
To each vertex $v\in \mathcal{K}\setminus S$ and each vertex $u\in \mathcal{I}$ orient the edge $uv$ towards $v$, in case it exists; and take an arbitrary transitive orientation of the edges in $G'[\mathcal{K}\setminus S]$.
Observe that $D$ is proper as the indegrees of vertices in $\mathcal{K}$ belong to $\{0, \dotsc, k-1, k+2, \dotsc, k'\}$ (recall that $G$ is cubic) and the indegrees of vertices in $\mathcal{I}$ belong to $\{k, k+1\}$ (as $S$ is a vertex cover).

Reciprocally, assume now that $G'$ admits a proper $k'$-orientation $D$.
Thus, observe that the indegree $d^-_D(v)$ of each $v\in\mathcal{I}$ must be at least~$k-1$ in any proper $k'$-orientation $D$ of $G'$, due to the subgraphs $Z_i$.
Moreover, $d^-_D(v)$ cannot be~$k-1$, due to the pendant $F(k-1,k')$ at $v$ and \cref{prop:Fik}.
Thus, in $D$ there must exist a vertex~$u\in \mathcal{K}$ such that the edge $uv$ is oriented towards $v$.

Since the vertices in $\mathcal{K}$ cannot have indegrees in $\{k, k+1\}$, due to the pendant subgraphs $F(i,k')$, for each~$i\in\{k, k+1\}$, it follows that each indegree in~$[k']\setminus \{k, k+1\}$ must be assigned to exactly one vertex of $\mathcal{K}$.
This implies that there exist a set $S$ on $k$ vertices in $\mathcal{K}$ with indegrees in $\{0\} \cup [k-1]$. 
Hence, all edges from vertices of $S$ to the vertices of $V(G') \setminus S$ must be oriented from $S$ towards $V(G')\setminus S$ and that all edges from a vertex $u$ in $\mathcal{I}$ to a vertex~$v$ in $\mathcal{K}\setminus S$ must be oriented towards $v$.
Thus, each vertex $u \in \mathcal{I}$ must be adjacent to a vertex of $S$, which implies that $S$ is a vertex cover of $G$ of cardinality $k$.
\end{proof}

\subsection{Remarks when the value of the solution is a parameter}
\label{sec:parameterized}

Let us comment a few words about the Parameterized Complexity of determining the proper orientation number of a graph when  parameterized by the value of the solution. Let us formally define the problem:

\paraProblemDef{Parameterized Proper Orientation}{A graph $G$ and a positive integer $k$.}{$k$}{Is $\po(G) \leq k$?}{pro:PPO}

It is important to recall that \nameref{pro:PPO} is para-\NP-complete in the class of planar bipartite graphs~\cite{ACD+15}. On the other hand, we can use previous results to deduce that this problem is \FPT~when restricted to chordal graphs.

\begin{proposition}
\label{prop:kernel_chordal}
\pname{Parameterized Proper Orientation} can be solved in time $\mathcal{O}(2^{k^2}\cdot k^{3k+1}\cdot n)$, for any chordal graph~$G$ on $n$ vertices.
\end{proposition}
\begin{proof}
 The algorithm presented in~\cite{ALSSS19} decides whether $\po(G)\leq k$, when~$G$ has treewidth at most $t$, in time $\mathcal{O}(2^{t^2}\cdot k^{3t}\cdot t \cdot n)$.
 Thus, given a chordal graph $G$, one may compute $\omega(G)$ in polynomial time, as they are perfect. In case $\omega(G)\geq k+2$, then the answer to \pname{Parameterized Proper Orientation} is trivially ``NO''. Otherwise, note that a chordal graph with $\omega(G)\leq k+1$ is a graph with treewidth at most $k$. Then, one can use the algorithm in~\cite{ALSSS19} to solve \pname{Parameterized Proper Orientation} in time $\mathcal{O}(2^{k^2}\cdot k^{3k+1}\cdot n)$.
\end{proof}

It is well known that if there exists an $f(k)\cdot\poly(n)$ algorithm to solve a problem $\Pi$ parameterized by $k$, then~$\Pi$ admits a kernel of size $f(k)$~\cite{Cyganetal2015}. Thus, \pname{Parameterized Proper Orientation} has a kernel of size $\mathcal{O}(2^{k^2}\cdot k^{3k})$ in the class of chordal graphs.
In the sequel, we argue that a kernel of size polynomial on $k$ is unlikely to exist. For this, we need the following definition and result. 

\begin{definition}[AND-cross-composition]
Let $L\subseteq \Sigma^*$ be a language and $Q \subseteq \Sigma^*\times \mathbb{N}$ be a parameterized language. We say that $L$ AND-cross-composes into $Q$ if there exists a polynomial equivalence relation ${\cal R}$
and an algorithm~${\cal A}$, called a cross-composition, satisfying the following conditions. The algorithm ${\cal A}$ takes as input a sequence of strings $x_1, x_2,\ldots , x_t \in \Sigma^*$
that are equivalent with respect to ${\cal R}$, runs in time
polynomial in $\sum_{i=1}^t |x_i|$, and outputs one instance $(y, k)\in \Sigma^*\times \mathbb{N}$ such
that:

\begin{enumerate}
    \item $k\leq p(\max^t_{i=1} |x_i| + \log t)$ for some polynomial function $p(\cdot)$, and
\item $(y, k) \in Q$ if and only if $x_i\in L$, for every $i\in\{1,\ldots,t\}$.
\end{enumerate}
\end{definition}

\begin{theorem}[\cite{FLSZ2019,Drucker2012}]\label{thm:andcross}
Assume that an \NP-hard language $L$ AND-cross-composes into a parameterized language~$Q$. Then $Q$ does not admit a polynomial compression, unless \NP $\subseteq$ \coNP/ \poly.
\end{theorem} 

Now, given two graphs $G$ and $H$, the \emph{disjoint union}, or simply \emph{union}, of $G$ and $H$ is the graph $G \cup H$ such that~$V(G \cup H) = V(G) \cup V(H)$ and $E(G \cup H) = E(G) \cup E(H)$.
Although the problem is not monotone for subgraphs even in trees~\cite{AHLS16}, one can observe that:
\begin{proposition}
\label{prop:disj_union}
If $G = G_1\cup G_2$, then $\po(G) = \max\{\po(G_1),\po(G_2)\}$.
\end{proposition}

By just taking the disjoint union of $t$ instances of \pname{Proper Orientation} in chordal graphs with the same $k$ (this is the equivalence relation ${\cal R}$), we have an AND-cross-composition to an instance of \pname{Parameterized Proper Orientation} in chordal graphs as the disjoint union of chordal graphs is chordal and \cref{prop:disj_union} holds. By \cref{thm:andcross}, it follows that:

\begin{corollary}
\pname{Parameterized Proper Orientation} restricted to chordal graphs does not admit polynomial kernel, unless \NP $\subseteq$ \coNP/ \poly.
\end{corollary}

In~\cite{ACD+15} the authors show that determining whether $\po(G) = \omega(G)-1$ can be done in polynomial time for a split graph~$G$, but even the problem of deciding whether~$\po(G)\leq\omega(G)$ remains open~\cite{ACD+15}. They also argue about the computational complexity of determining $\po(G)$ when $G$ is cobipartite.
In case determining whether $\po(G)\leq k$ turns out \NP-complete for split graphs and cobipartite graphs, then a natural question is to wonder about its parameterized version. Note that split graphs are chordal and thus Proposition~\ref{prop:kernel_chordal} can be applied. Below, we give a better kernel for split graphs, and provide a linear kernel for cobipartite graphs.

\begin{proposition}
\pname{Parameterized Proper Orientation} admits a $\mathcal{O}(2^kk^2)$-kernel for split graphs.
\end{proposition}
\begin{proof}
Let $(G,k)$ be an instance of \pname{Parameterized Proper Orientation} such that $G$ is split. Let $\{K,S\}$ be a partition of $V(G)$ such that $K$ is a maximum clique of $G$ and $S$ is a stable set. By \cref{prop:disj_union}, we can suppose that $G$ is connected. Write $\omega = \omega(G)$, and let $M = k\omega - \binom{\omega}{2}$. 
We construct, in polynomial time, an equivalent instance $(G',k)$ such that $|G'|\leq (k+1)+2^{k+1} (M+1) = \mathcal{O}(2^kk^2)$. 

If $|K|=\omega\geq k+2$, then the answer is ``NO'' (and thus one can just output $G' = K_{k+2}$ in this case). Now, we assume that $\omega\leq k+1$. Two vertices $u$ and $v$ in a graph $G$ are \emph{twins} if $N_G(u)=N_G(v)$. We partition the vertices of~$S$ into at most $2^\omega-1$ subsets $S_i$ of twins, for $i\in\{1,\ldots, 2^\omega\}$.

\begin{claim}
\label{claim:reduction_rule_split}
If there exists $i\in\{1,\ldots, 2^\omega\}$ such that $|S_i|\geq M+2$, then removing $(|S_i|-(M+1))$ vertices from $S_i$ builds a smaller graph $G'$ such that $\po(G)\leq k$ if, and only if, $\po(G')\leq k$.
\end{claim}
\begin{claimproof}
Note that if $D$ is a proper $k$-orientation of $G$ such that exactly $\ell$ edges with one endpoint in $K$ and the other in $S$ are oriented towards $K$, then:
\[k\geq \max_{v\in K} d_D^-(v) \geq \frac{\sum_{v\in K} d_D^-(v)}{\omega} = \frac{\binom{\omega}{2}+\ell}{\omega} \Rightarrow \ell\leq k\omega-\binom{\omega}{2} = M.\]
Since $|S_i|\geq M+2$, note that at most $M$ vertices of $S_i$ are tails of arcs pointing towards their (common) neighbors in $K$, and thus all the other $|S_i|-M\geq 2$ vertices of $S_i$ are sinks in $D$. Thus, they all have the same indegree in $D$ (equal to their degree) and the same neighborhood (since they are twins). Thus, the removal of $(|S_i|-(M+1))$ of these vertices leads to $G'$ and note that the orientation provided by $G$ is still proper in $G'$. Recall that the problem is not monotone and thus this step is indeed necessary.

On the other hand, if $D'$ is a proper $k$-orientation of $G'$, since we left $M+1$ vertices in $S_i$, at least one of these vertices, say $s_i^*$ must be a sink in $D'$, by the previous analysis. Consequently, no vertex in $K\cap N(s_i^*)$ may have indegree equal to $d_{D'}^-(s_i^*) = d_G(s_i^*)$. Thus, we may extend such orientation of $D'$ into an orientation $D$ of $G$ by just letting the removed vertices be sinks as well in $D$.
\end{claimproof}

By applying Claim~\ref{claim:reduction_rule_split} to each feasible $S_i$, $i\in\{1,\ldots,2^\omega\}$, we build a graph $G'$ as required.
\end{proof}

\begin{proposition}
\pname{Parameterized Proper Orientation} admits a linear kernel for cobipartite graphs.
\end{proposition}
\begin{proof}
The proof is straightforward as either $\omega(G)\geq k+2$, in which case the answer is ``NO'' (and one can output the equivalent instance $(K_{k+2},k)$), or we have that $|V(G)|\leq 2(k+1)$.
\end{proof}

\subsection{Linear-time algorithm on quasi-threshold graphs}
\label{sec:quasi_threshold}

The \emph{join} of $G$ and $H$ is the graph $G \wedge H$, such that $V(G \wedge H) = V(G) \cup V(H)$ and $E(G \wedge H) = E(G) \cup E(H) \cup \{uv \mid u \in V(G),~v \in V(H)\}$.
With a slight abuse of notation, we refer to $v\wedge H$ to be the join of the trivial graph $(\{v\},\emptyset)$ and~$H$. A \emph{quasi-threshold graph}~\cite{JINGHO96} is isomorphic to a~$K_1$, or can be formed by a finite sequence of two operations: the disjoint union of two quasi-threshold graphs or the join of a new vertex and a quasi-threshold graph. This decomposition theorem can be applied to easily prove the following result.


\begin{theorem}
\label{thm:triviallyPG}
If $G$ is a quasi-threshold graph, then $\po(G) = \omega(G)-1$ and $\po(G)$ can be computed in linear time.
\end{theorem}
\begin{proof}
It suffices to prove that $\po(G) = \omega(G)-1$, since the second part follows by the fact that $G$ is a chordal graph. By induction on the number~$n$ of vertices of $G$. 
It is trivial that~$\po(G)=0$ for~$n=1$.
If~$G=(H_1 \cup H_2)$, then, by \cref{prop:disj_union}, $\po(G) = \max\{\po(H_1), \ \po(H_2)\}=\max\{\omega(H_1)-1,\ \omega(H_2)-1\} = \max\{\omega(H_1),\ \omega(H_2)\}-1 = \omega(G)-1$. 
If~$G=(v \wedge H)$, then~$\omega(G) = \omega(H)+1$, which implies that~$\po(G)\geq\omega(G)-1 = \omega(H) = \po(H)+1$.
Let $D$ be an optimal proper orientation of~$H$.
Let~$D'$ be an extension of~$D$ to~$G$ by orienting all edges~$vu$ towards~$u$.
Then all indegrees in~$D$ are increased by exactly~1 and the indegree of~$v$ is~0 in~$D'$.
Hence~$\po(G)\leq \po(H)+1=\omega(G)-1$ and the theorem follows.
\end{proof}

\section{Bounds}

\subsection{Split graphs}
\label{sec:split}

In this section, we prove that~$\po(G)$ is linearly upper bounded by~$\omega(G)$, when $G$ is split and we present a tight example.

\begin{theorem}
\label{thm:upperboundsplit}
Let $G$ be a split graph.
Then, $\po(G) \le 2\omega(G)-2$.
\end{theorem}
\begin{proof}
Let $G=(V,E)$ be a split graph partitioned into a maximal clique $\mathcal{K}$ and an independent set $\mathcal{I}$.
Note that~$\omega(G) = |\mathcal{K}|$.
From now on, we will use~$\omega = \omega(G)$.

Let $K_h \subseteq \mathcal{K}$ be the subset of $\mathcal{K}$ of size $h$ such that $K_h = \{v \in \mathcal{K} \mid d(v) \geq 2\omega-2\}$.
Note that each vertex in $K_h$ has $\omega-1$ neighbors in $\mathcal{K}$ and at least $\omega-1$ neighbors in $\mathcal{I}$.
Let $(v_1, \dotsc, v_h)$ be any ordering of $K_h$ and let $E_i$ be an arbitrary set of $\omega-1$ edges connecting $v_i$ to vertices of $\mathcal{I}$, for each $v_i \in K_h$.
Let us build a proper $(2\omega-2)$-orientation of $G$. 

First, orient all the edges in $[\mathcal{K}\setminus K_h,K_h]$ towards $K_h$ and orient all edges in $E_i\cup\{v_jv_i\mid 1\le j < i\le h\}$ towards~$v_i$, for each $i\in \{2,\dotsc,h\}$.
If $h<\omega$, orient the edges in $E_1$ towards $v_1$ and we explain later how to complete such orientation.
Otherwise, orient all the remaining non-oriented edges towards $\mathcal{I}$.
Note that, in both cases, we have that $d^-(v_i) = 2\omega - 2 - (h-i)$, for each $i\in \{2,\dotsc,h\}$.

If $h=\omega$, the orientation is already completely described.
Observe that, in this case, $d^-(v_1)=0$, $\omega\le d^-(v_i)\le 2\omega-2$ for every $i\in \{2,\dotsc,h\}$, and $d^-(u)\le d(u)\le \omega-1$ for every $u\in \mathcal{I}$.
Thus, we have the desired proper $(2\omega-2)$-orientation in this case.

Otherwise, if $h<\omega$, then $d^-(v_1)=2\omega-1-h$, while every vertex of $\mathcal{I}$ has indegree 0 so far.
This means that we have a partial proper $(2\omega-2)$-orientation, which we extend as follows.
Let $G'$ be the graph obtained by removing the edges that have already been oriented and let $(u_1,\dotsc,u_q)$ be an ordering over the vertices in $V(G)\setminus K_h$ such that, for every $i\in [q]$, we have $d_{G'_i}(u_i)=\Delta(G'_i)$, where $G'_i = G'\setminus \{u_{i+1},\dotsc,u_q\}$. To complete such partial orientation, orient all the edges of $G'$ from $u_i$ to $u_j$, for every $1\le i< j\le q$, and all the remaining non-oriented edges in $[K_h,\mathcal{I}]$ towards $\mathcal{I}$.
We prove that the obtained orientation is a proper $(2\omega-2)$-orientation of $G$. For this, we prove that for every $i\in [q]$:
\begin{enumerate}
    \item $d^-(u_i)\neq d^-(u_j)$, for every $u_j\in N(u_i)$; and
    \item $d^-(u_i)\le2\omega-2-h$.
\end{enumerate}  Observe that, because $2\omega-1-h\le d^-(v_i)\le 2\omega-2$ for every $i\in [h]$, it follows that the orientation is a proper $(2\omega-2)$-orientation. 

To see that (1) holds, suppose, without loss of generality that $i<j$, which means that $u_iu_j$ is oriented towards~$u_j$.
We then have that $d^-(u_i)\le d_{G'_j}(u_i)-1< d_{G'_j}(u_j) = d^-(u_j)$.
To see that (2) holds, observe that if~$u_i\in \mathcal{I}$, then $d^-(u_i)\le d_G(u_i) \le \omega-1\le 2\omega-2-h$ since $h\le \omega-1$.
On the other hand, if $u_i\in \mathcal{K}\setminus K_h$, then $d^-(u_i)\le d_{G'}(u_i) = d_G(u_i)-h< 2\omega-2-h$ by the construction of $K_h$.
\end{proof}

Next, we prove that the upper bound on \cref{thm:upperboundsplit} is tight.
The tight example is obtained from a clique $K$ of size $\omega$ by adding several independent sets having as neighbors the same subset of vertices of~$K$. 
The idea is to increase the average indegree of $K$ in any proper orientation, which leads to the lower bound.  



\begin{theorem}
\label{thm:tightexamplesplit}
For every positive integer $\omega$, there exists a split graph $G$ with $\omega(G) = \omega$ such that~$\po(G) = 2\omega-2.$
\end{theorem}
\begin{proof}
The proof is done by constructing a split graph $G$, with $V(G) = (\mathcal{K}, \mathcal{I})$, from a complete graph~$\mathcal{K} = K_{\omega}$ and taking some independent set that forces a vertex to have high indegree.

Given an integer $k \in [w]$, let $\mathcal{I}_k$ be an independent set having $\binom{\omega}{k}$ vertices, each connected to a different set of~$k$ vertices in $\mathcal{K}$.
In a proper orientation of our graph, either every arc leaves $\mathcal{K}$ and goes into $\mathcal{I}_k$ or at least one arc from $\mathcal{I}_k$ goes into $\mathcal{K}$.
If every arc goes into $\mathcal{I}_k$, then the indegree of every vertex in $\mathcal{I}_k$ is $k$, therefore no vertex in~$\mathcal{K}$ could have indegree $k$.
Otherwise, there is at least one arc coming from $\mathcal{I}_k$ and into~$\mathcal{K}$.

To construct our split graph $G$, let $\mathcal{I}$ be a set composed of $\omega(\omega-1)$ copies of $\mathcal{I}_k$, for every $k \in [\omega-1]$.
It holds that $\mathcal{I}$ is an independent set, since the neighborhood of each~$\mathcal{I}_k$ is $\mathcal{K}$, and that $\omega(G) = \omega$, since no vertex in $\mathcal{I}$ is adjacent to every vertex in~$\mathcal{K}$.
Consider a proper orientation of the graph $G$.

If no vertex in $\mathcal{K}$ has indegree in the set $[\omega-1]$, then, for every vertex $v \in \mathcal{K}$, it holds that $d^{-}(v) \in \{0\} \cup \{ \omega, \omega+1, \omega+2, \dotsc\}$.
Since $\mathcal{K}$ is complete, then every indegree in $\mathcal{K}$ must be unique, which, by the Pigeonhole Principle, means that there is at least one vertex $u$ such that $d^{-}(u) \ge 2\omega - 2$.

Otherwise, let $k \in [\omega-1]$ be the indegree of some vertex in $\mathcal{K}$.
Then, by the construction of $G$, there exists at least one edge oriented towards $\mathcal{K}$ from each copy of $\mathcal{I}_k$, which means that:
\[ \sum_{v \in \mathcal{K}} d^{-}(v) \ge |E(\mathcal{K})| + \omega(\omega-1) = \dfrac{3}{2}\omega(\omega -1). \]

Notice that if $\mathcal{K}$ did not have any vertex with indegree higher than $2\omega - 3$, then:
\[ \sum_{v \in \mathcal{K}} d^{-}(v) \le \sum_{j = 0}^{\omega-1} (\omega-2 + j) = \dfrac{(3\omega-5)\omega}{2}. \]

Comparing these two expressions, we conclude that there must exist a vertex $u \in \mathcal{K}$ such that $d^{-}(u) \ge 2\omega - 2$.

Therefore, in a proper orientation of $G$, we have $\overrightarrow{\chi}(G) \ge 2\omega - 2$.
By \cref{thm:upperboundsplit}, we deduce that~$\po(G) = 2\omega -2$.
\end{proof}

\subsection{Uniform block graphs}
\label{sec:block}

Let us first prove a general result that we need in the sequel.

\begin{proposition}
\label{prop:extending_partial_orientations}
Given a graph $G$, a subset $S\subseteq V(G)$ and a proper $k$-orientation $D_S$ of $G[S]$ such that, for each $v\in V(G)\setminus S$ and $u\in N(v)\cap S$, we have $|N(v)\cap S| > d_{D_S}^-(u)$, then $D_S$ can be extended into a proper $\Delta(G)$-orientation $D$ of $G$ such that:
\begin{itemize}
    \item if $v\in S$, then $d_D^-(v) = d_{D_S}^-(v)$; and
    \item if $v\in V(G)\setminus S$, then $d_D^-(v)\leq d_G(v)$. 
\end{itemize}
\end{proposition}
\begin{proof}
Let $\overline{S} = V(G)\setminus S$. 
We start constructing a proper orientation $D$ of $G$ by orienting all edges in $G[S]$ as in the orientation $D_S$ and by orienting all edges in $[S,\overline{S}]$ toward their endpoints in $\overline{S}$. Let $D_0$ be such partial orientation of $G$.

Note that no two adjacent vertices $u,v\in S$ may have equal indegrees in $D$, since $D$ is an extension of $D_0$ and all their incident edges have already been oriented, we have that $d_{D}^-(u)=d_{D_0}^-(u)=d_{D_S}^-(u)\neq d_{D_S}^-(v) = d_{D_0}^-(v)=d_{D}^-(v)$. By hypothesis, we cannot have two vertices $u\in S$ and $v\in \overline{S}$ with equal indegree in $D$ either as $d_{D}^-(u)=d_{D_0}^-(u)=d_{D_S}^-(u) < |N(v)\cap S|= d_{D_0}^-(v)\leq d_{D}^-(v).$ 

Thus, we just need to care about orienting the edges in $G[\overline{S}]$ while ensuring that no two vertices $u,v\in \overline{S}$ have the same indegree. In fact, notice that if a vertex $u\in\overline{S}$ is such that all edges incident to $u$ are already oriented in $D_0$, then $u$ has no neighbors in $\overline{S}$. Thus, we just need to worry about vertices of $\overline{S}$ having at least one neighbor in $\overline{S}$. 

To do this, for each $i\in\{1,\ldots,|\overline{S}|\}$, we choose a vertex $v_i \in\overline{S}$ incident to some non-oriented edge in $D_{i-1}$, and we orient all the non-oriented edges incident to $v_i$ toward $v_i$, thus obtaining  a partial orientation $D_i$. We then let $D = D_{|S|}$. Let us define in which order to choose such vertices.

For each vertex $v\in \overline{S}$, we define $m_i(v)$ to be the number of edges incident to $v$ that are non-oriented in the partial orientation $D_{i-1}$ and we define its \emph{potential indegree at step $i$} as $\pot_i(v) = d_{D_0}^-(v) + m_i(v)$. We take $v_i$ to be such that $m_i(v_i)\geq 1$ and $\pot(v_i) = \max \{\pot(u)\mid u\in \overline{S}\wedge m_i(u)\geq 1\}$. As previously said, orientation $D_i$ is then obtained from $D_{i-1}$ by orienting all the $m_i(v_i)$ non-oriented edges incident to $v_i$ at $D_{i-1}$ toward $v_i$. 

We claim that $D=D_{|\overline{S}|}$ is a proper orientation of $G$. By contradiction, suppose that there are two adjacent vertices $u$ and $v$ such that $d_D^-(u)=d_D^-(v)$. As we argued, this is only possible if $u,v\in \overline{S}$ and $u$ and $v$ had incident non-oriented edges in $D_0$. Without loss of generality, suppose that $u$ is chosen before $v$ at step $i$ to have its incident non-oriented edges to be oriented towards $u$. In particular, the edge $uv$ is oriented towards $u$. Thus, observe that $d_D^-(u) = \pot_i(u)\geq \pot_i(v)>\pot_{i+1}(v)\geq d_D^-(v)$, a contradiction.
\end{proof}

As an easy corollary, we get the following.

\begin{corollary}
\label{cor:source_surplus}
For any graph $G$ and $u\in V(G)$, there exists a proper $\Delta(G)$-orientation $D$ of $G$ such that $d^-_D(u)=0$.
\end{corollary}
\begin{proof}
It suffices to apply Proposition~\ref{prop:extending_partial_orientations} when $S=\{u\}$.
\end{proof}

A \emph{block graph} is a graph whose 2-connected components are complete graphs. 
A block graph $G$ is \emph{$k$-uniform} if every block of $G$ has size $k= \omega(G)$.

In the next definition there is a slight abuse of notation as we consider that the first color of a proper coloring of a graph $G$ might be zero. Let $G$ be any graph, $u\in V(G)$, and $c,d$ be positive integers. We say that an orientation $D$ of $G$ is a \emph{proper orientation compensated by $(c,d,u)$} if $d^-_D(u)=d$ and the following is a proper coloring of $G$:

\[\fc{D}{u}{c}(v) = \left\{\begin{array}{ll}
d^-_D(v) & \text{, if } v\neq u\\
c      & \text{, otherwise}
\end{array}\right.\]

A block is called a \emph{leaf block} if it contains exactly one cut-vertex of $G$. In what follows, given a $k$-uniform block graph $G$, we show how to obtain a proper $(3k-2)$-orientation of $G$ by removing blocks in the vicinity of the leaf blocks, orienting the remaining graph, then extending the orientation to the removed blocks. For this, we treat smaller cases first. 

We say that a block graph $G$ is a \emph{path block graph} if its block cut-point graph is a path. Note that it thus consists of a sequence of cliques $C_1,\dotsc,C_q$ such that $C_i\cap C_{i+1}\neq C_{i+1}\cap C_{i+2}$ for every $i\in [q-2]$. Denote by $u_i$ the vertex in $C_i\cap C_{i+1}$. 

\begin{lemma}
Given $k\geq 3$, let $G$ be a $k$-uniform path block graph with clique sequence $C_1, \dotsc, C_q$, let $u\in C_q$ be a non-cut-vertex, and $c,d$ be positive integers such that 
either $c > k - 1\geq d$ or $c = d = k-1$. Then, there exists a proper $\ell$-orientation of $G$ compensated by $(c,d,u)$, where  $\ell=\max\{c,2k-2\}$.\label{lem:pathblock}

\end{lemma}
\begin{proof}
By induction on $q$. If $q=1$, then let $D$ be any transitive orientation of $G$ corresponding to a total ordering $(v_1,\dotsc, v_n)$ of $V(G)$ such that $v_{d+1} = u$. Thus, $d^-_D(v_i) = i-1$ for each $i\in [n]$, and in particular $d^-_D(u)=d$ as required. Consequently, $\fc{D}{u}{c}(v_i) = i-1$, for every $i\in [n]\setminus \{d+1\}$. Note that, regardless $c > k - 1\geq d$ or $c = d = k-1$, we have that $D$ is a proper $c$-orientation of $G$ compensated by $(c,d,u)$.

Now suppose $q\ge 2$ and let $G'$ be obtained from $G$ by removing the vertices in $C_q\setminus \{u_{q-1}\}$. 
In case $d\geq 1$, let $D'$ be a proper $(2k-2)$-orientation of $G'$ in which $u_{q-1}$ has indegree zero (it exists, by Corollary~\ref{cor:source_surplus}). Let $D$ be an orientation of $G$ obtained from $D'$ by taking a transitive orientation of $C_q$ in such a way that $u_{q-1}$ is the source and $u$ has in-degree exactly $d$. Note that $D$ is a proper $\ell$-orientation of $G$ compensated by $(c,d,u)$.

Assume then that $d=0$ (which implies that $c>k-1\geq 1$, by hypothesis). In case $c\neq 2k-2$, let $D'$ be a proper $\ell$-orientation of $G'$ compensated by $(2k-2,k-1,u_{q-1})$. Let $D$ be an orientation of $G$ obtained from $D'$ by taking a transitive orientation of $C_q$ in such a way that $u_{q-1}$ is the sink and $u$ is the source. Again, since $c\neq 2k-2$ and since $D'$ is a proper $\ell$-orientation of $G'$ compensated by $(2k-2,k-1,u_{q-1})$, note that $D$ is a proper $\ell$-orientation of $G$ compensated by $(c,d,u)$. Finally, suppose that $c=2k-2$. Analogously, let $D'$ be a proper $\ell$-orientation of $G'$ compensated by $(2k-3,k-1,u_{q-1})$. Let $D$ be an orientation of $G$ obtained from $D'$ by taking a transitive orientation of $C_q$ in such a way that $u_{q-1}$ has indegree $k-2$ and $u$ is the source (recall that $k\geq 3$). Once more, one may deduce that $D$ is a proper $\ell$-orientation of $G$ compensated by $(c,d,u)$.
\end{proof}


Let us now introduce some notations that we use in the next key lemmas. Consider that $G$ is a connected $k$-uniform block graph, for $k\geq 3$. Recall that its block-cutpoint graph is a tree $T$. With a slight abuse of notation, when we refer to a block $B$, we also mean the vertex in $T$ corresponding to $B$.
Thus, we consider that $G$ is rooted at some block $R$ (it means rooting $T$ at the vertex corresponding to $R$). Given a cut-vertex $u$, let $B_u$ be the parent block of $u$ (similarly, the vertex corresponding to $B_u$ is the parent of $u$ in the block-cutpoint tree rooted at $R$), and $B^1_u,\dotsc, B^q_u$ be its children blocks (defined analogously). We denote by $H^i_u$ the component of $G-u$ containing $B^i_u\setminus\{u\}$, and by $G^i_u$ the subgraph of $G$ induced by $V(H^i_u)\cup \{u\}$. We say that $u$ is a \emph{path connector} if $H^i_u$ is a path block graph for every $i\in [q]$. We also call each $B_u^i$ a \emph{path block}, i.e. the vertices in $B_u^i$ together with the vertices in its descendant blocks induce a path block subgraph.

As we prove ou main result by contradiction, we say that a $k$-uniform block graph $G$ for $k\geq 3$ is a \emph{minimum counter-example} (MCE for short) if $G$ is a connected $k$-uniform block graph such $\po(G)>3k-2$ having $p$ blocks, but any connected $k$-uniform block graph $H$ with less than $p$ blocks satisfy $\po(G)\leq 3k-2$.

\begin{lemma}\label{lem:cleanPaths}
If $G$ is a rooted MCE, then each path connector $u$ of $G$ has at most two children blocks.
\end{lemma}
\begin{proof}
By contradiction, assume that $u$ is a path connector of $G$ having $q\geq 3$ children blocks. Because $G$ is $k$-uniform, note that $d_G(u)\geq 4k-4>3k-3$, as $k\geq 3$.

Let $G' = G - \bigcup_{i=1}^q V(H^i_u)$. Since $G$ is an MCE, let $D'$ be a proper $(3k-2)$-orientation of $G'$. In the sequel, we show how to extend $D'$ to a proper $(3k-2)$-orientation $D$ of $G$, contradicting that $\po(G)>3k-2$. 

If $d^-_{D'}(u) = 0$, then to each $G_u^i$ we apply Corollary~\ref{cor:source_surplus} and obtain an a proper orientation $D^i$ of $G_u^i$ in which $u$ is a source, for each $i\in [q]$. Since $\Delta(G-V(G'))\le 2k-2$, each $D^i$ is a proper $(2k-2)$-orientation of $G_u^i$. Thus, if we extend $D'$ by using $D^i$ in each $G_u^i$ we obtain a proper $(3k-2)$-orientation of $G$, a contradiction.

Thus, assume that $d^-_{D'}(u) \geq 1$. Let $F =\{d^-_{D'}(v)\mid v\in B_u\setminus \{u\}\}$ be the set of forbidden indegrees to $u$. Let $c\in \{2k-2,\dotsc,3k-2\}\setminus F$. Note that $c$ exists since $|F|\le k-1$. We want to ensure that we can extend $D'$ to a proper $(3k-2)$-orientation of $G$ in such a way that $u$ has indegree $c$. 

Define $d'= c-d^-_{D'}(u)$ as the \emph{missing indegree for $u$}. Observe that $k-1\le d' \le 3k-3$ (recall that $1\le d^-_{D'}(u)\le k-1)$. Let $d_1 = k-1$, $d_2=\min\{k-1,d'-(k-1)\}$, $d_3 = \max\{0,d'-(2k-2)\}$ and $d_i=0$, for each $i\in\{4,\ldots, q\}$. By Lemma~\ref{lem:pathblock}, let $D^i$ be a proper $(3k-2)$-orientation of $G_u^i$ compensated by $(c,d_i,u)$, for each $i\in[q]$. Again, if we extend $D'$ by using $D^i$ in each $G_u^i$ we obtain a proper $(3k-2)$-orientation of $G$, a contradiction.
\end{proof}

Let us add some notation. If $G$ is a rooted $k$-uniform block graph, $k\geq 3$, formed by a root block $B$ with cut-vertices $u_1,\dotsc,u_q$ and such that each $u_i$ is a path connector, then we say that $G$ is a \emph{crossroad block graph}, and we call block $B$ the \emph{cross-point block}. Now, if $u$ is a cut-vertex in a rooted $k$-uniform block graph $G$ that is not a path connector and such that every children block of $u$ is either a path block or a cross-point block, then we say that $u$ is a \emph{cluster cut-vertex}. 

\begin{lemma}\label{lem:cleanClusters}
If $G$ is a rooted MCE, then $G$ has no cluster cut-vertex.
\end{lemma}
\begin{proof}
By contradiction, let $u$ be a cluster cut-vertex of $G$ having $p\geq 1$ children blocks $B_u^1,\ldots, B_u^p$ and parent block $B_u$. Since $u$ is a cluster cut-vertex, at least one of its children blocks must be a cross-point block. Without loss of generality, assume that $B_u^1$ is a cross-point block and let $u_1,\ldots,u_q$ be the cut-vertices of $B_u^1$ distinct from $u$. Note that $1\leq q\leq k-1$. Moreover, since each $u_i$ is a path connector, by Lemma~\ref{lem:cleanPaths} we get that $u_i$ has at most two children blocks $B_{u_i}^1$ and $B_{u_i}^2$ and thus $d_G(u_i)\leq 3k-3$.

\medskip

\noindent{\bf Case $p=1$:} let $G^1 = G - V(H^1_u)$. Since $G$ is an MCE, let $D^1$ be a proper $(3k-2)$-orientation of $G^1$. Let us show how to extend $D^1$ into a proper $(3k-2)$-orientation $D$ of $G$, contradicting the hypothesis that $\po(G)>3k-2$. 

Note that $d_{D^1}^-(u)\leq k-1 = d_{G^1}(u)$.
In order to obtain $D$ from $D^1$, first take a transitive orientation $D_B$ of $B_u^1$ such that $u$ is the source and $u_i$ satisfies $d_{D_B}^-(u_i) = i$, for each $i\in\{1,\ldots, q\}$. Then, let $D_i^{1}$ be a proper $(2k-2)$-orientation of $G_{u_i}^{1}$ compensated by $(k-1+i,k-1,u_i)$; and, if $B_{u_i}^2$ exists, let $D_i^{2}$ be a proper $(2k-2)$-orientation of $G_{u_i}^{2}$ compensated by $(k-1+i,0,u_i)$, otherwise let $D_i^{2}=\emptyset$.

Let $D$ be the orientation of $G$ obtained by the union of $D^1$, $D_B$ and $D_i^{j}$, for every $i\in\{1,\ldots, q\}$ and each $j\in\{1,2\}$. Note that $D$ is a proper $(3k-2)$-orientation of $G$ as $D_B$ is transitive, $d_D^-(u) = d_{D^1}^-(u) \leq k-1$, $k\leq d_D^-(u_i) = k-1+i\leq 2k-2$ and Lemma~\ref{lem:pathblock} holds.
\medskip

\noindent{\bf Case $p = 2$:} In case $B_u^1$ and $B_u^2$ are both cross-point blocks, then we can use the same idea as in case $p=1$ twice as we may extend a proper $(3k-2)$-orientation of $G-V(H_u^1)-V(H_u^2)$ to $G_u^1$ and then to $G_u^2$.

Suppose then that $B_u^2$ is a path block. Let $D^2$ be a proper $(3k-2)$-orientation of $G^2=G-V(H_u^1)$. In case $d_{D^2}^-(u)\notin\{k,\ldots, 2k-2\}$, then again we can use the same idea as in case $p=1$ to extend $D^2$ into a proper $(3k-2)$-orientation $D$ of $G$.

Otherwise, note that since $d_{G^2}(u)=2k-2$, then $u$ is the tail of $r\geq 1$ arcs whose heads are in $B_u^2$. In order to obtain $D$, we first remove the orientation of all edges in $B_u^2$. Then, we take a transitive orientation $D_B$ of $B_u^1$ such that $u_1$ is a source and $d_{D_B}^-(u)=r\geq 1$. Let $S=\{k,\ldots,2k-2\}\setminus\{d_{D^2}^-(u)\}$. Note that $S$ has $k-2$ values that we now use to set the indegrees of $u_2,\ldots, u_q$ ($u_1$ will have indegree $k-1$). Let $g:\{u_2,\ldots,u_q\}\to S$ be any injective function. Now, using Lemma~\ref{lem:pathblock}, we take a proper $(2k-2)$-orientations:

\begin{itemize}
    \item $D_{u_1}^1$ of $G_{u_1}^1$ compensated by $(k-1,k-1,u_1)$;
    \item $D_{u_1}^2=\emptyset$ if $G_{u_1}^2$ is empty, or compensated by $(k-1,0,u_1)$ otherwise;
    \item $D_{u_i}^1$ of $G_{u_i}^1$ compensated by $(g(u_i), g(u_i)-d_{D_B}^-(u_i),u_i)$, for each $i\in\{2,\ldots, q\}$;
    \item $D_{u_i}^2=\emptyset$ if $G_{u_i}^2$ is empty, or compensated by $(g(u_i), 0,u_i)$, for each $i\in\{2,\ldots, q\}$.
\end{itemize}
Finally, take any proper $(2k-2)$-orientation $D_u^2$ of $G_u^2$ compensated by $(d_{D^2}^-(u), 0, u)$. Note that the orientation $D$ obtained by the union of all such orientations is a proper $(3k-2)$-orientation of $G$ because $d_D^-(u) = d_{D^2}^-(u)$, $D_B$ is transitive, $g$ is injective, $d_D^-(u)$ does not belong to the image of $g$ and Lemma~\ref{lem:pathblock} holds.

\medskip
\noindent{\bf Case $p\geq 3$:} let $G^3 = G - \bigcup_{i=1}^3 V(H^i_u)$. Since $G$ is an MCE, let $D^3$ be a proper $(3k-2)$-orientation of $G^3$. Again, we extend $D^3$ into a proper $(3k-2)$-orientation $D$ of $G$, contradicting the hypothesis that $\po(G)>3k-2$.

In case $d_{D^3}^-(u)=0$, then we claim that we can extend $D^3$ to $G$ by using Corollary~\ref{cor:source_surplus}. In fact, the subgraph $G^*$ of $G$ induced by $V(G_u^1)\cup V(G_u^2)\cup V(G_u^3)$ has maximum degree $3k-3$ thanks to Lemma~\ref{lem:cleanPaths}. Thus, by Corollary~\ref{cor:source_surplus} we can obtain a proper $(3k-3)$-orientation $D^*$ of $G^*$ in which $u$ is a source. By joining $D^*$ to $D^3$, we obtain a proper $(3k-2)$-orientation of $G$.

Thus, assume that $d_{D^3}^-(u)\geq 1$. We first increase the indegree of $u$ to ease the process of extending $D^3$. Let $F=\{d_{D^3}(w)\mid (u\neq w)\text{ and }w\in B_u\}$ be the forbidden indegrees to $u$ by its neighbors in $B_u$. Note that $|F|=k-1$ and thus there exists $c\in \{2k-1,\ldots, 3k-2\}\setminus F$. Let $b_1,b_2,b_3$ be integers in $\{0,\cdots,k-1\}$ such that $(d_{D^3}^-(u)+b_1+b_2+b_3)$ is equal to $c$.

To obtain $D$, we first take transitive orientations $D_u^i$ of $B_u^i$ in such a way that $u$ has indegree $b_i$, for each $i\in\{1,2,3\}$. Thus, we set the indegree of $u$ to be $c = d_{D}^-(u)+b_1+b_2+b_3\geq 2k-1$. Now, since the indegree of $u$ is greater than $2k-2$, one can easily apply Lemma~\ref{lem:pathblock} to find orientations to the path blocks that are descendant of each $B_u^i$, for each $i\in\{1,2,3\}$.
\end{proof}

We are finally ready to prove the main result of this section.

\begin{theorem}
Let $G$ be a $k$-uniform block graph. Then $\po(G)\le 3k-2$.
\end{theorem}
\begin{proof}
Since $\po(T)\leq 4$ for any tree $T$~\cite{ACD+15,Knoxetal2017}, and $\po(G)\leq 7$ for any cactus graph $G$~\cite{AHLS16}, we assume that $k\geq 4$. 

By contradiction, let $G$ be an MCE having $q$ blocks. Note that $q\geq 4$ as otherwise $\po(G)\leq\Delta(G)\leq 3k-3$. Furthermore, $G$ cannot be a path block graph as otherwise $\po(G)\leq \Delta(G)\le 2k-2$.

Consider $G$ to be rooted in some block $R$. Since $G$ is not a path block graph, there must be a cut-vertex $u$ of $G$  whose children blocks are either cross-point blocks or path blocks.
However, this contradicts Lemma~\ref{lem:cleanClusters} as $u$ is a cluster cut-vertex.
\end{proof}

Now we deal with the subclass of connected $k$-uniform block graphs in which each block contains at most two cut-vertices, $k\ge 3$. 

\begin{proposition}
\label{prop:particular_k_uniform_block}
If $G$ is a connected $k$-uniform block graph such that $k\geq 3$ and each block contains at most two cut-vertices, then $\po(G)\leq k+1$.
\end{proposition}
\begin{proof}
Let us consider the block cut-point tree $T$ of $G$. Recall that $V(T)$ has a vertex $v_C$ corresponding to each maximal clique $C$ of $G$ and a vertex $u_x$ corresponding to each cut-vertex $x$ of $G$. Also, a vertex $v_C$ is adjacent to $u_x$ in $T$ if, and only if, $x\in C$. We further consider that the descendants of each vertex of $T$ are ordered (thus, we refer to the leftmost descendant as being the first in such ordering).

The idea is to build a proper $(k+1)$-orientation $D$ of $G$ by just using transitive orientations of the maximal cliques of $G$, while ensuring that the cut-vertices $x$ have in-degree $d_D^-(x)\in\{0,k,k+1\}$. We claim that it is enough to prove that there are no two adjacent cut-vertices with in-degree $d\in\{k,k-1\}$. Indeed, this is true since: vertices of in-degree zero are sources (implying that their neighbors have in-degree at least one); the orientation of each maximal clique is transitive; and vertices that are not cut-vertices have degree in $G$, and thus in-degree in $D$, at most $k-1$.
%

If $G$ has no cut-vertices, then $G$ is complete and $\po(G)\leq \Delta(G)=\omega(G)-1$. Otherwise, we root $T$ at a leaf node $v_R$, and let $r$ be the cut-vertex of $R$. We start building $D$ by orienting $R$ with any transitive orientation in which $r$ is a source.

Now, we consider a more general setup in which we have a partial orientation $D$ of $G$, we have to choose a transitive orientation of a clique $C'$ and we know that the clique $C$ in the path between $C'$ and $R$ is already oriented. 
Let $x$ be the common cut-vertex of $C$ and $C'$. During the construction of $D$, we will ensure $x$ has at most two arcs of $C$ oriented towards $x$. Thus, we choose the orientation of $C'$ according to the following case analysis.

If $x$ is the unique cut-vertex of $C'$, i.e. $C'$ is a leaf block of $G$, then we orient $C'$ with any transitive orientation satisfying that $x$ is a source in $C'$. Otherwise, let $x'$ be the other cut-vertex of $C'$ (recall that, by hypothesis, there are at most two of them).

If $x$ is a source of $C$, then orient $C'$ in such a way that $x$ is also a source of $C'$ and $x'$ has in-degree one in $C'$. Note that $x$ will have in-degree zero in $D$ in this case.

Consider now the case in which $x$ has in-degree one in $C$. If $v_{C'}$ is the leftmost child of $u_x$ in $T$, then we choose any transitive orientation of $C'$ such that $x$ is a sink in $C'$ and $x'$ is a source in $C'$. Otherwise, take any transitive orientation of $C'$ in which $x$ is a source in $C'$ and $x'$ has in-degree two in $C'$ (note that this is possible since $k\geq 3$). Observe that, since the leftmost child of $u_x$ is always oriented in a way that $x$ is a sink, we get that $x$ will have in-degree $k$ in $D$.

Finally, consider now the case in which $x$ has in-degree two in $C$. We proceed analogously to the previous case. If $v_{C'}$ is the leftmost child of $u_x$ in $T$, then we choose any transitive orientation of $C'$ such that $x$ is a sink in $C'$ and $x'$ is a source in $C'$. Otherwise, take any transitive orientation of $C'$ in which $x$ is a source in $C'$ and $x'$ has in-degree one in $C'$. Similarly as before, observe that in this case $x$ will have in-degree $k+1$ in $D$.

By construction, note that we cannot have two adjacent cut-vertices with the same in-degree $d\in \{k,k+1\}$.
\end{proof}

Let us now provide a tight example for Proposition~\ref{prop:particular_k_uniform_block}. Note that, unlike Proposition~\ref{prop:particular_k_uniform_block}, this example also holds for trees. Actually, Proposition~\ref{prop:tightexample_block2cutpoints} generalizes the example  of a tree having proper orientation number 3 that appeared in~\cite{ACD+15}. 

\begin{proposition}
\label{prop:tightexample_block2cutpoints}
For every $k \geq 2$, there exists a $k$-uniform block graph such that each block contains at most two cut-vertices and $\po(G)\geq k+1$.
\end{proposition}
\begin{proof}
For a given positive integer $k \ge 2$, let us build $G(k)$ as follows. We start with a clique and pick two distinct vertices, $u$ and $v$. On each, we add $k$-cliques, denoted by $K_0(u), \ldots, K_{k}(u)$ whose common vertex is only $u$ and $K_0(v), \ldots, K_{k}(v)$, whose common vertex is only $v$. For each clique $K_i(w)$, $w\in \{u,v\}$, $i \ge 1$, pick another vertex $w_i \ne w$ to which we attach another $k$-clique. We claim that $\po(G(k)) \ge k+1$.

Suppose, by contradiction, that $\po(G) \le k$ and let $D$ be a proper $k$-orientation of $G(k)$. By symmetry, we assume that the edge $uv$ is oriented towards $u$.

Let us first prove that $d_D^{-}(u) = k$. Suppose that $d_D^{-}(u) = \ell \le k-1$. Notice that every vertex $w \in K_0 \setminus \{u\}$ satisfy $d_{G(k)}(w) = k-1$. We conclude that $\{d_D^{-}(w) \mid w \in K_0\} = \{0, 1, \ldots, k-1\}$, which means that $u$ must receive all its $\ell$ arcs from $K_0$. It is a contradiction, since $(v,u) \in A(D)$.

Thus, $d_D^{-}(u) = k$. Since $vu \in D$, then for at least one of the $k$ cliques $K_1(u), \ldots K_k(u)$, let us say $K_i(u)$, $u$ acts like a source, meaning that $uu_i \in D$. Using an analogous proof as the one for $u$, we conclude that $d_D^{-}(u_i) = k$, which is a contradiction, since $uu_i \in E(G(k))$. Therefore, $\po(G(k)) \ge k+1$.
\end{proof}

Thanks to Proposition~\ref{prop:kernel_chordal} 
and to the fact that trees have proper orientation number at most 4~\cite{ACD+15}, we know that there is a polynomial-time algorithm to compute $\po(T)$, for a tree $T$. However, the structure of trees satisfying $\po(T)=i$, for $i\in\{2,3,4\}$, is not known. In~\cite{ACD+15}, the authors presents two classes of trees $T$ satisfying $\po(T)\leq i$, for $i=2$ and $i=3$. Next, we present a result that reveals another class of trees satisfying such inequalities thanks to Proposition~\ref{prop:extending_partial_orientations}.

\begin{corollary}
\label{cor:no_adjacent_high_degree}
Let $G$ be a graph and $c$ be any positive integer. If $G$ has no two adjacent vertices of degree at least $c+1$, then $\po(G)\leq c$.
\end{corollary}
\begin{proof}
Let $S$ be the set of vertices of degree at least $c+1$. Since $S$ is an independent set, we have that $\po(G[S])=0$ and we can apply Proposition~\ref{prop:extending_partial_orientations} to obtain the desired proper $c$-orientation of $G$.
\end{proof}

Corollary~\ref{cor:no_adjacent_high_degree} indicates new families of trees and (outer)planar graphs having bounded proper orientation number. It also implies in particular that:

\begin{corollary}
\label{cor:upperbound_trees}
If $T$ is a tree having no two adjacent vertices of degree at least three, then $\po(T)\leq 2$. If $T$ is a tree having two no two adjacent vertices of degree at least four, then $\po(T)\leq 3$.
\end{corollary}



\subsection{Maximal outerplanar graphs}
\label{sec:maximal_outerplanar}

A simple graph $G$ is \emph{maximal outerplanar} if it is outerplanar and no edge can be added to $G$ so that the obtained graph is still outerplanar. A \emph{maximal outerplane} graph $G$ is a planar embedding of a maximal outerplanar graph. Note that, by definition, the outerface of a maximal outerplane graph $G$ is the only face that might not be of degree three and thus these graphs are chordal.

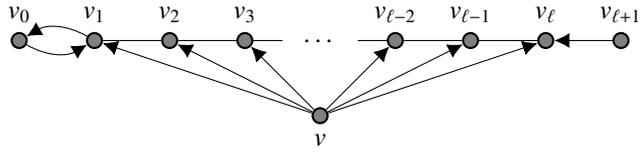
\begin{figure}[!ht]
    \centering
    
    \begin{tikzpicture}
    	
    	    \tikzset{node/.style={draw, fill=white, circle, thick, fill = white, auto=left, inner sep=2pt}}
            \tikzset{rectangle/.append style={draw,rounded corners,dashed, minimum height=0.4cm}}
            \tikzset{>={Latex[width=2mm,length=2mm]}}
            
            \node[node] (v) at (3,-1) [fill=gray, label=below:$v$] {};
            
    	    \node[node] (v0) at (-1,0) [fill=gray, label=above:$v_0$] {};
    	    
    	    \node[node] (v1) at (0,0) [fill=gray, label=above:$v_1$] {};
    	    
    	    \node[node] (v2) at (1,0) [fill=gray, label=above:$v_2$] {};
    	    
    	    \node[node] (v3) at (2,0) [fill=gray, label=above:$v_3$] {};
    	    
    	    \node[node] (vdm2) at (4,0) [fill=gray, label=above:$v_{\ell-2}$] {};
    	    
    	    \node[node] (vdm1) at (5,0) [fill=gray, label=above:$v_{\ell-1}$] {};
    	    
    	    \node[node] (vd) at (6,0) [fill=gray, label=above:$v_{\ell}$] {};
    	    
    	    \node[node] (vdp1) at (7,0) [fill=gray, label=above:$v_{\ell+1}$] {};
    	    
    	    \node (dots) at (3,0) {$\dotsc$};

		    \path[-] (v1) edge node {} (v2)
            (v2) edge node {} (v3)
            (vdm2) edge node {} (vdm1)
            (vdm1) edge node {} (vd)
            (v3) edge node {} (2.5,0)
            (vdm2) edge node {} (3.5,0);
            
            \path[->] 
            (v) edge node {} (v1)
            (v) edge node {} (v2)
            (v) edge node {} (v3)
            (v) edge node {} (vdm2)
            (v) edge node {} (vdm1)
            (v) edge node {} (vd)
            (vdp1) edge node {} (vd);
            
            \draw [->] (v1) to [bend right] (v0);
            \draw [->] (v0) to [bend right] (v1);
            
	    \end{tikzpicture}
    \caption{Partial orientations to be extended by \cref{lem:extending_to_path}.}
    \label{fig:lemma_outerplanar}
\end{figure}

\begin{lemma}
\label{lem:extending_to_path}
Given an integer $\ell\geq 6$ and a graph $G$, let $v, v_0,\ldots, v_{\ell+1}\in V(G)$ such that $P = (v_1,\ldots, v_\ell)$ is a path in $G$ and $N_G(v_i)=\{v_{i-1},v,v_{i+1}\}$, for each $i\in[\ell]$.

Let $D_1$ and $D_2$ be proper $k$-orientations of $G-E(P)$, for some $k\geq 3$ such that $(v,v_i)\in A(D_j)$, $d^-_{D_1}(v_1)= 2$, $d^-_{D_2}(v_1)= 1$,  $d^-_{D_j}(v_\ell)=2$ and $d^-_{D_j}(v)\geq 4$, for each $i\in[\ell]$ and for each $j\in\{1,2\}$, (see \cref{fig:lemma_outerplanar}). Then, $D_1$ and $D_2$ can be extended into proper $k$-orientations of $G$.
\end{lemma}
\begin{proof}
Note that since $d^-_{D_j}(v)\geq 4$ we cannot have a conflict between $v$ and one of its neighbors $v_i$, for any $i\in[\ell]$, because in any extension $D'_j$ of $D_j$ to $G$, we have $d^-_{D'_j}(v_i)\leq d(v_i)=3<4\leq d^-_{D_j}(v)$, for each $j\in\{1,2\}$. Thus, we just need to consider some case analysis with respect to the possible indegrees of $v_0$ and $v_{\ell+1}$ in $D_j$.

In what follows, we use the possible \emph{alternating} orientations of $P$ according to the parity of $\ell$. If $\ell$ is odd, then such path admits two orientations with no two consecutive arcs in the same direction. One orientation in which $v_1$ and $v_\ell$ are sinks in the path, to which we refer as $D^{sk}(P)$ and another in which both are sources, to which we refer as $D^{sr}(P)$. In case $\ell$ is even, we consider the two possible alternating orientations as well, one in which $v_1$ is a source and $v_\ell$ is a sink, which we name $D^{1\to}(P)$, and the other in which the opposite occurs and we name $D^{\to 1}(P)$. Moreover, for each $1\leq i\leq j\leq \ell$, we denote the subpath $(v_i,\ldots, v_j)$ of $P$ by $P_i^j$.

Let us first argue how to extend the orientation $D_1$ to $G$.
In case $\ell$ is even, we can extend $D_1$ to $G$ by using one of the orientations $D^{1\to}(P)$ or $D^{\to 1}(P)$, unless $d_{D_1}^-(v_0)=d_{D_1}^-(v_{\ell+1})=2$ or $d_{D_1}^-(v_0)=d_{D_1}^-(v_{\ell+1})=3$. If $d_{D_1}^-(v_0)=d_{D_1}^-(v_{\ell+1})=2$, we orient the edge $v_1v_2$ toward $v_1$ and orient the subpath $P_2^{\ell}$ with the orientation $D^{sk}(P_2^{\ell})$.
In case $d_{D_1}^-(v_0)=d_{D_1}^-(v_{\ell+1})=3$, we orient the edges $v_1v_2$ and $v_2v_3$ toward $v_2$, the edges $v_{\ell-2}v_{\ell-1}$ and~$v_{\ell-1}v_{\ell}$ toward~$v_{\ell-1}$ and orient the subpath $P_3^{\ell-2}$ with the orientation $D^{1\to}(P_3^{\ell-2})$ (note that since $\ell$ is even, $P_3^{\ell-2}$ has an even number of vertices and we need this path to contain at least one edge - this is true because $\ell\geq 6$).

Let us consider then the case in which $\ell$ is odd. In this case, one can use one of the orientations $D^{sk}(P)$ or~$D^{sr}(P)$ over the path $P$ to obtain a proper $k$-orientation $D$ of $G$, unless $\{d^-_{D_1}(v_0),d^-_{D_1}(v_{\ell+1})\} = \{2,3\}$. By symmetry, we assume without loss of generality that $d_{D_1}^-(v_0)=2$ and $d^-_{D_1}(v_{\ell+1}) = 3$. We orient $v_1v_2$ towards $v_1$ and complete the orientation $D_1$ with the orientation $D^{\to 1}(P_2^{\ell})$.

Let us consider now how to extend the orientation $D_2$. Note that the only difference is that the edge $v_0v_1$ is oriented towards $v_0$. Consider first that $\ell$ is even.
In case $d^-_{D_2}(v_0)\neq 1$ and $d^-_{D_2}(v_{\ell+1})\neq 3$, we can use the orientation~$D^{1\to}(P)$ to extend $D_2$ to $G$.
In case $d^-_{D_2}(v_0)\neq 2$ and $d^-_{D_2}(v_{\ell+1})\neq 2$, we can use the orientation $D^{\to 1}(P)$ to extend $D_2$ to $G$.
In case $d^-_{D_2}(v_0) = 1$ and $d^-_{D_2}(v_{\ell+1})= 2$, we orient the edge $v_{\ell-1}v_\ell$ towards $v_\ell$ and we orient the path $P_1^{\ell-1}$ with the orientation $D^{sk}(P_1^{\ell-1})$. In case $d^-_{D_2}(v_0) = 2$ and $d^-_{D_2}(v_{\ell+1})= 3$, we orient the edges $v_{\ell-2}v_{\ell-1}$ and~$v_{\ell-1}v_\ell$ towards $v_{\ell-1}$ and we orient $P_1^{\ell-2}$ with the orientation $D^{1\to}(P_1^{\ell-2})$.

Finally, suppose that $\ell$ is odd. Similarly, we can use one of the orientations $D^{sr}(P)$ or $D^{sk}(P)$ to extend $D_2$ to $G$ unless $d^-_{D_2}(v_0) = d^-_{D_2}(v_{\ell+1})= 2$; or $d^-_{D_2}(v_0) = 1$ and $d^-_{D_2}(v_{\ell+1})= 3$. In case~$d^-_{D_2}(v_0)  = d^-_{D_2}(v_{\ell+1})= 2$, then we orient $v_{\ell-1}v_\ell$ towards $v_\ell$ and we orient the path $P_1^{\ell-1}$ with the orientation~$D^{1\to}(P_1^{\ell-1})$. In case $d^-_{D_2}(v_0) = 1$ and $d^-_{D_2}(v_{\ell+1})= 3$, we orient the edges $v_{\ell-2}v_{\ell-1}$ and $v_{\ell-1}v_\ell$ towards $v_{\ell-1}$ and we orient $P_1^{\ell-2}$ with the orientation $D^{sk}(P_1^{\ell-2})$.  
\end{proof}

\begin{theorem}
\label{thm:outerplanar_upbound}
If $G$ is a maximal outerplane graph whose weak-dual is a path, then $\po(G)\leq 13$.
\end{theorem}
\begin{proof}
We prove the theorem by induction in the number of vertices of $G$.
If $\Delta(G)\leq 13$, then we have nothing to prove as $\po(G)\leq \Delta(G)\leq 13$.

Thus, let $v$ be a vertex of degree $\Delta = \Delta(G)\geq 14$. Let the neighbors of $v$ in clockwise order be $v_1,\ldots, v_{\Delta}$ so that the edge $vv_1$, as well as $vv_\Delta$, lie in the outerface (see \cref{fig:outerplanar}). Note that since the weak-dual of $G$ is a path, then the vertices $v_3,\ldots, v_{\Delta-2}$ have degree three.

Let $G_1$ and $G_2$ be the maximal outerplane graphs obtained from $G$ by removing the vertices $v, v_3,\ldots, v_{\Delta-2}$ (note that $G_1$ or $G_2$ might be composed of a single edge). By induction hypothesis, $G_1$ and $G_2$ admit proper 13-orientations $D_1$ and $D_2$, respectively.

We construct a proper 13-orientation $D$ of $G$ from $D_1$ and $D_2$ 
by keeping the vertices~$v_1$, $v_2$, $v_{\Delta-1}$, $v_{\Delta}$ with the same indegrees provided by $D_1$ and $D_2$. For this, we orient all edges~$v_1v$, $v_2v$, $v_{\Delta-1}v$ and $v_{\Delta}v$ toward $v$ and we orient $v_2v_3$ and $v_{\Delta-1}v_{\Delta-2}$ toward $v_3$ and $v_{\Delta-2}$, respectively, as depicted in \cref{fig:outerplanar}. Thus, in the sequel we show how to orient the remaining edges incident to $v_3,\ldots, v_{\Delta-2}$ in order to obtain a proper 13-orientation of $G$.

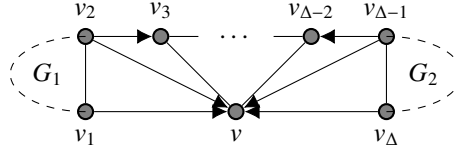
\begin{figure}[!ht]
    \centering
    
    \begin{tikzpicture}
    	
    	    \tikzset{node/.style={draw, fill=white, circle, thick, fill = white, auto=left, inner sep=2pt}}
            \tikzset{rectangle/.append style={draw,rounded corners,dashed, minimum height=0.4cm}}
            \tikzset{>={Latex[width=2mm,length=2mm]}}
            
    	    \node[node] (v) at (0,0) [fill=gray, label=below:$v$] {};
    	    \node[node] (v1) at (-2,0) [fill=gray, label=below:$v_1$] {};
    	    
    	    \node[node] (v2) at (-2,1) [fill=gray, label=above:$v_2$] {};
    	    
    	    \node[node] (v3) at (-1,1) [fill=gray, label=above:$v_3$] {};
    	    
    	    \node[node] (vdm2) at (1,1) [fill=gray, label=above:$v_{\Delta-2}$] {};
    	    
    	    \node[node] (vdm1) at (2,1) [fill=gray, label=above:$v_{\Delta-1}$] {};
    	    
    	    \node[node] (vd) at (2,0) [fill=gray, label=below:$v_{\Delta}$] {};
    	    
    	    \node (dots) at (0,1) {$\dotsc$};
    	    
    	    \node (g1) at (-2.5,0.5) {$G_1$};
    	    \node (g2) at (2.5,0.5) {$G_2$};

		    \path[<-] (v) edge node {} (v1)
            (v) edge node {} (v2)
            (v) edge node {} (vdm1)
            (v) edge node {} (vd)
            (v3) edge node {} (v2)
            (vdm2) edge node {} (vdm1);
            
            \path[-] 
            (v3) edge node {} (v)
            (v3) edge node {} (-0.5,1)
            (vdm2) edge node {} (0.5,1)
            (vdm2) edge node {} (v);
            
            \draw[dashed] (v2) arc(90:270:1 and 0.5);
            \draw[dashed] (vd) arc(-90:90:1 and 0.5);
            
            \draw (v1) -- (v2) (vd) -- (vdm1);
	    \end{tikzpicture}
    \caption{A representation of a maximal outerplanar $G$ and a vertex of maximum degree $v\in V(G)$.}
    \label{fig:outerplanar}
\end{figure}

Because the edges $v_1v$, $v_2v$, $v_{\Delta-1}v$ and $v_{\Delta}v$ are oriented toward $v$ note that $d^-_D(v)\geq 4$.
In order to ensure that $v$ has no conflict with its neighbors in $\{v_1,v_2,v_{\Delta-1},v_{\Delta}\}$, we orient the least number edges in $(v_3v,v_4v,v_5v,v_6v)$, following this order, towards $v$. It means that if $4\notin \{d^-_{D_1}(v_1),d^-_{D_1}(v_2), d^-_{D_2}(v_{\Delta-1}), d^-_{D_2}(v_{\Delta})\} = S$, then none of these edges is oriented towards $v$. In case $4\in S$, but $5\notin S$, then we only orient the edge $v_3v$ towards $v$, and so on.
All edges~$v_iv$, for every $i\in\{7,\ldots, \Delta-1\}$, are oriented towards $v_i$. Let us now argue how to orient the edges of the path $P=(v_3,v_4,\ldots, v_{\Delta-2})$ to complete the orientation $D$ while ensuring that it is proper.

In case $4\notin S$, as we said, we set $d_D^-(v)=4$ by orienting all edges $v_3v,\ldots, v_{\Delta-1}v$ toward $v_3,\ldots, v_{\Delta-1}$, respectively. Thus, we can extend such orientation to the edges of the path $P$, by \cref{lem:extending_to_path} (note that $v_2$ corresponds to~$v_0$ in \cref{lem:extending_to_path}, and in this case we know that the edge $v_2v_3$ is oriented towards $v_3$).

Suppose now that $4\in S$, but $5\notin S$. In this case, we orient $v_3v$ towards $v$ and we orient $v_iv$ toward $v_i$, for each $i\in\{4,\ldots,\Delta-2\}$. Again, it remains to orient the edges of the path $P$. First, we choose an orientation of $v_3v_4$ to extend the partial orientation we have so far in a way that the indegree of $v_3$ is distinct from $v_2$. Then, now it remains to orient the edges of the path $P_4^{\Delta-2}$ and again we may apply \cref{lem:extending_to_path} (note that $v_3$ corresponds to $v_0$ in \cref{lem:extending_to_path}, but now the edge $v_3v_4$ is not necessarily oriented towards $v_4$).

The remaining cases can be solved analogously. Notice that, in the last case, all edges $v_3v$, $v_4v$, $v_5v$ and~$v_6v$ are oriented towards $v$. Then, after extending the partial orientation, step by step, to the vertices $v_3$, $v_4$, $v_5$ and~$v_6$, it remains to orient the path $(v_7,\ldots, v_{\Delta-2})$. Since $\Delta\geq 14$, this path has at least 6 vertices as required by \cref{lem:extending_to_path}.
\end{proof}

We emphasize that we just wanted to prove that there is a constant upper bound to the proper orientation number of the graphs in this class. This upper bound can probably be improved by a more careful case analysis.

\subsection{Claw-free chordal graphs}
\label{sec:claw_free}

Given two graphs $G$ and $H$, we say that $G$ is \emph{$H$-free} if $H$ is not an induced subgraph of $G$. A \emph{claw} is the connected graph on four vertices having three of them with degree one.
The following proposition is quite straightforward. 

\begin{proposition}
\label{prop:bound_chordal_claw_free}
    If $G$ is a chordal claw-free graph, then $\po(G)\leq \Delta(G)\leq 3\omega(G)$.
\end{proposition}
\begin{proof}
We prove that the neighborhood of any vertex $v\in V(G)$ can be partitioned into at most three cliques.
If~$N(v)$ is a clique or $|N(v)|\leq 3$ we have nothing to prove. So let $u,w\in N(v)$ be such that $uw\notin E(G)$ and we may consider that there are at least two other vertices in $N(v)$ distinct from $u$ and $w$.

Since $G$ is claw-free, note that any vertex $z\in N(v)\setminus\{u,w\}$ must be adjacent to at least one of the vertices in~$\{u,w\}$. Denote by $N_u^* = (N(v)\cap N(u))\setminus N(w)$. Again, since $G$ is claw-free, observe that $N_u^* \cup\{u\}$ must be a clique, as otherwise one may use $w$ and $v$ to find an induced claw in $G$. Analogously, one may define $N_w^*$ and observe that $N_w^* \cup \{w\}$ must also be a clique. Finally, vertices in $N_{uw}=N(v)\setminus (\{u,w\}\cup N_u^* \cup N_w^*)$ must be adjacent to both $u$ and $w$. Thus, now we use the hypothesis that $G$ is chordal to deduce that $N_{uw}$ must be a clique, as otherwise one may find a chordless cycle on four vertices using $u, w$ and two non-adjacent vertices of $N_{uw}$.
\end{proof}

\subsection{Cographs}
\label{sec:cographs}

It is well known that cographs are exactly the $P_4$-free graphs. Thus, they are $C_k$-free graphs, for each $k\geq 5$, but they may contain chordless $C_4$'s and thus they are not necessarily chordal.

The class of cographs is recursively defined as follows: $K_1$ is a cograph; if $G_1$ and $G_2$ are cographs, then $G = G_1\wedge G_2$ is a cograph; if $G_1$ and $G_2$ are cographs, then $G = G_1\cup G_2$ is a cograph.

If $G$ is a $K_1$, then $\po(G) = 0$. The case of disjoint union is solved by Proposition~\ref{prop:disj_union}. However, computing the proper orientation number of a cograph which is the join of two other cographs $G_1$ and $G_2$ - even if we know optimal solutions for $G_1$ and $G_2$ - seems not trivial. Here, we provide some bounds. Recall that $\Ad(G) = \frac{|E(G)|}{|V(G)|}$.

\begin{proposition}
\label{prop:bound_cographs}
Let $G_1,G_2$ be graphs on $n_1,n_2$ vertices, respectively, and let $G = G_1\wedge G_2$. Then,
\[\min\left\{\Ad(G_1)+\frac{n_2}{2},\Ad(G_2)+\frac{n_1}{2}\right\}\leq \po(G)\leq \min\left\{ \po(G_1)+n(G_2),\po(G_2)+n(G_1)\right\}.\]
\end{proposition}
\begin{proof}
To prove the upper bound, it suffices to take proper $\po(G_i)$-orientations of $G_i$, for each $i\in\{1,2\}$, and either orient all edges of $[V(G_1),V(G_2)]$ toward their endpoint in $V(G_1)$ or all edges toward their endpoint in $V(G_2)$.
Note that such an orientation $D$ is proper as either $d_ D^-(v_1)\leq \po(G_1)\leq \Delta(G_1)\leq n(G_1)-1$ and $n(G_1)\leq d_ D^-(v_2)\leq n(G_1)+\po(G_2)$, for each $v_1\in V(G_1)$ and each $v_2\in V(G_2)$; or the analogous holds when we orient all edges in $[V(G_1),V(G_2)]$ toward $G_1$.

Now, let $m_i = |E(G_i)|$, for each $i\in\{1,2\}$. To prove the lower bound, since $G$ is the join of two graphs $G_1$ and $G_2$, at least half of the edges in $[V(G_1),V(G_2)]$ must be oriented towards $G_1$ or towards $G_2$. Let $D$ be a proper $\po(G)$-orientation of $G$ and suppose that these edges are oriented towards $G_2$ (the other case is analogous). Then, we have that:
\[\po(G) =\max_{v\in V(G)} d_D^-(v)\geq \max_{v_2\in V(G_2)} d_D^-(v_2) \geq \frac{\sum_{v_2\in V(G_2)} d_D^-(v_2)}{n_2} \geq \left(m_2+\frac{n_1\cdot n_2}{2}\right)\cdot \frac{1}{n_2} = \Ad(G_2)+\frac{n_1}{2}.\]
\end{proof}

Note that the upper bound is tight when $G$ is complete.

\section{Further Research}
\label{sec:further}

We provided upper bounds for split graphs and cographs. As we mentioned, the computational complexity of determining $\po(G)$ when $G$ is split or cobipartite is still open~\cite{ACD+15}. We further propose the following problem:

\begin{problem}
What is the computational complexity of determining $\po(G)$, when $G$ is a cograph?
\end{problem}

As discussed in Section~\ref{sec:chordal}, a $\mathcal{O}(2^{k^2} \cdot k^{3k+1})$-kernel exist to \pname{Parameterized Proper Orientation} when restricted to chordal graphs. One could try to improve such kernel, although, as we commented, no polynomial kernel may exist, unless ~{\NP$ \ \subseteq$ \ \coNP/\poly}.

Concerning bounds, we believe that finding a constant upper bound for (outer)planar graphs, a question posed in~\cite{ACD+15}, is one of the most challenging open problems about this parameter. 

We provided upper bounds to some subclasses of chordal graphs $G$ in function of $\omega(G)$. Our bound for split graphs and our remark on cobipartite graphs answer Problem 9 in~\cite{ACD+15}. We suggest the following related question:

\begin{problem}
Does there exist a constant $c$ such that $\po(G)\leq c \cdot\omega(G)$, for any chordal graph $G$?
\end{problem}

In particular, we show that such constant exist for the case of $k$-uniform block graphs. Since the problem is not monotone, one may wonder whether $\po(G)\leq 3 \omega(G) -2 $ holds for every block graph $G$.



\bibliographystyle{plain}
\bibliography{proper-chordal}

\end{document}